\documentclass[journal]{IEEEtran}
\IEEEoverridecommandlockouts
\usepackage{amssymb}
\usepackage{amsmath}
\usepackage{amsthm}
\usepackage{graphicx}
\usepackage{comment}
\usepackage{booktabs}
\usepackage[nocompress]{cite}
\usepackage{bm}
\usepackage{float}
\usepackage{multirow}
\usepackage{color}
\usepackage{subfigure}
\usepackage{caption}
\usepackage{epstopdf}
\usepackage{hyperref}
\usepackage{fixltx2e}
\usepackage{longtable}
\usepackage{diagbox}
\usepackage{changepage}
\usepackage{cases}
\usepackage{optidef}
\usepackage{makecell}
\usepackage[linesnumbered,ruled]{algorithm2e}
\theoremstyle{definition}

\theoremstyle{remark}
\newtheorem{remark}{Remark}
\newtheorem{assumption}{Assumption}

\theoremstyle{plain}
\newtheorem{theorem}{Theorem}
\newtheorem{lemma}{Lemma}

\def\BibTeX{{\rm B\kern-.05em{\sc i\kern-.025em b}\kern-.08em
    T\kern-.1667em\lower.7ex\hbox{E}\kern-.125emX}}
\begin{document}

\title{Privacy-preserving Intelligent Resource Allocation for Federated Edge Learning in Quantum Internet
\author{Minrui~Xu, Dusit~Niyato,~\emph{Fellow,~IEEE}, Zhaohui~Yang, Zehui~Xiong, Jiawen~Kang*,\\ Dong~In~Kim,~\emph{Fellow,~IEEE}, and Xuemin~(Sherman)~Shen,~\emph{Fellow,~IEEE}}
\thanks{Minrui~Xu and Dusit~Niyato are with the School of Computer Science and Engineering, Nanyang Technological University, Singapore (e-mail: minrui001@e.ntu.edu.sg; dniyato@ntu.edu.sg); Zhaohui~Yang is with the College of Information Science and Electronic Engineering, Zhejiang University, Hangzhou 310007, China, and Zhejiang Provincial Key Lab of Information Processing, Communication and Networking (IPCAN), Hangzhou 310007, China, and also with Zhejiang Laboratory, Hangzhou 31121, China. (e-mail: zhaohuiyang92@gmail.com); Zehui~Xiong is with the Pillar of Information Systems Technology and Design, Singapore University of Technology and Design, Singapore 487372, Singapore (e-mail: zehui\_xiong@sutd.edu.sg); Jiawen~Kang is with the School of Automation, Guangdong University of Technology, China (e-mail: kavinkang@gdut.edu.cn). Dong~In~Kim is with the Department of Electrical and Computer Engineering, Sungkyunkwan University, Suwon 16419, South Korea (e-mail: dikim@skku.ac.kr); Xuemin~(Sherman)~Shen is with the Department of Electrical and Computer Engineering, University of Waterloo, Waterloo, ON, Canada, N2L 3G1 (e-mail: sshen@uwaterloo.ca). (\textit{*Corresponding author: Jiawen Kang})}
}


\maketitle
\begin{abstract}
Federated edge learning (FEL) is a promising paradigm of distributed machine learning that can preserve data privacy while training the global model collaboratively. However, FEL is still facing model confidentiality issues due to eavesdropping risks of exchanging cryptographic keys through traditional encryption schemes.  Therefore, in this paper, we propose a hierarchical architecture for quantum-secured FEL systems with ideal security based on the quantum key distribution (QKD) to facilitate public key and model encryption against eavesdropping attacks. Specifically, we propose a stochastic resource allocation model for efficient QKD to encrypt FEL keys and models. In FEL systems, remote FEL workers are connected to cluster heads via quantum-secured channels to train an aggregated global model collaboratively. However, due to the unpredictable number of workers at each location, the demand for secret-key rates to support secure model transmission to the server is unpredictable. The proposed systems need to efficiently allocate limited QKD resources (i.e., wavelengths) such that the total cost is minimized in the presence of stochastic demand by formulating the optimization problem for the proposed architecture as a stochastic programming model. To this end, we propose a federated reinforcement learning-based resource allocation scheme to solve the proposed model without complete state information. The proposed scheme enables QKD managers and controllers to train a global QKD resource allocation policy while keeping their private experiences local. Numerical results demonstrate that the proposed schemes can successfully achieve the cost-minimizing objective under uncertain demand while improving the training efficiency by about 50\% compared to state-of-the-art schemes.
\end{abstract}


%
\IEEEpeerreviewmaketitle

\section{Introduction}
%
%
%
%
\IEEEPARstart{A}{rtificial} Intelligence (AI)  enables a wide range of computing and networking applications in edge networks, e.g., smart cities~\cite{wu2020collaborate, yang2021energy,choi2022enabling}, Internet of Vehicles~\cite{chekired20195g, cheng2019space}, and Metaverses~\cite{huynh2022artificial,xu2022quantum,xu2022full}. As one of the critical technologies in AI, federated edge learning (FEL) is a novel paradigm of privacy-preserving machine learning (ML) for intelligent edge networks~\cite{lim2020federated}. In FEL, multiple data owners (a.k.a., FEL workers) can train a global model collaboratively for a model owner without exposing their sensitive raw data.
To ensure the security of data and models in FEL systems, many modern cryptographic schemes are applied~\cite{huang2021starfl}, such as secure multi-party computation (MPC), trusted execution environment (TEE), and safe key distribution.  For instance, a secure and trusted collaborative edge learning framework is proposed in~\cite{tang2022secure}, where homomorphic encryption (HE) and blockchain are leveraged to track and choke malicious behaviors.
With the rapidly increasing computation power of quantum computers~\cite{xuan2021minimizing}, novel techniques will be brought to empower FEL systems, including large-scale searching, optimization and semantic communication. However, the FEL systems based on existing schemes are under serious security threats. For example, traditional key distribution schemes based on the hardness in computing of certain mathematical problems are no longer considered to be safe in the post-quantum era~\cite{huang2021starfl}. Fortunately, based on the quantum no-cloning theorem~\cite{buvzek1996quantum} and the Heisenberg’s uncertainty principle~\cite{weyl1927quantenmechanik}, quantum key distribution (QKD)~\cite{scarani2009security} is promising for providing proven secure key distribution schemes for collaborative training between FEL workers and model owners by facilitating public key and model encryption against eavesdropping attacks.

Originated from classical QKD schemes such as Bennett-Brassard-1984 (BB84)~\cite{bb84} and Grosshans-Grangier-2002 (GG02)~\cite{grosshans2002continuous}, some modern QKD schemes have paved the way for the Quantum Internet~\cite{wehner2018quantum} in recent years. For example, the measurement device-independent QKD (MDI-QKD)~\cite{lo2012measurement} provides one of the practical QKD solutions by increasing the range of secure communications and filling the detection gaps with an untrusted relay by avoiding any eavesdropping attacks on the Quantum Internet.
Although several existing works focus on the theoretical and experimental aspects of the deployments of MDI-QKD, the problems of QKD resource allocation in the Quantum Internet have been largely overlooked~\cite{cao2022evolution}. For example, in~\cite{cao2021hybrid}, a deterministic programming model and a heuristic approach based on the shortest path algorithm are proposed to optimize the deployment cost of QKD resources. However, the problem of optimal allocation of QKD resources for quantum-secured FEL systems with heterogeneous data and model owners remains open. In particular, the number of participating FEL workers at different locations and times is uncertain due to unpredictable node and device failures~\cite{lim2021decentralized}. Therefore, different security levels might be required by cluster heads to encrypt local models during global aggregation. Specifically, the secret-key rate for reaching the information-theoretic security (ITS) requirement is dynamic to support the encryption of intermediate model and related information according to uncertainties in quantum-secured FEL systems.


To address these uncertainty issues, we propose a stochastic QKD resource (i.e., wavelength) allocation model to optimize the QKD deployment cost of the Quantum Internet. To protect FEL models and public keys from eavesdropping attacks, we propose a hierarchical architecture for quantum-secured FEL systems that includes the FEL layer, the control and management layer, and the QKD infrastructure layer. To handle the dynamics of security demands from the FEL layer, we model the QKD resource allocation of QKD managers and QKD controllers in the control and management layer as a stochastic programming model that allocates QKD resources from the QKD infrastructure layer to cluster heads in the FEL layer. However, the proposed stochastic model can hardly be applied in practice because it requires complete state information from FEL nodes and QKD nodes, which is infeasible for QKD managers and controllers to collect due to privacy concerns~\cite{yang2019federated}. Fortunately, the independent QKD resource allocation problems can be addressed by the promising deep reinforcement learning (DRL) algorithms~\cite{li2020noma}. {Nevertheless, the efficiency and stability of learning-based approaches still face the issues of ``data islands" during their training and inference~\cite{yang2019federated}:
\begin{itemize}
    \item Initially, QKD managers and controllers configure QKD nodes to provide QKD resources based only on practical observation of the state of the FEL layer. However, the experiences, including observations, actions, and rewards, are kept in the local replay buffers due to privacy concerns. Therefore, the lack of collaboration between QKD managers and controllers makes QKD resource allocation problems more challenging to satisfy changing security demands in quantum-secured FEL systems.
    \item Furthermore, QKD managers are reluctant to share their rewards from the FEL layers directly with QKD controllers. Therefore, QKD controllers can only collect incomplete experiences, including states and actions, during their interaction with the FEL systems, which are insufficient for their local policy improvement. Therefore, QKD controllers can only use policies shared by QKD managers to instruct QKD resource allocation decisions independently to QKD nodes for the FEL layers.
\end{itemize}
These issues could lead to inadequate training efficiency and unstable inference performance for learning-based algorithms in privacy-preserving environments.}

To overcome the aforementioned issues, in this paper, we propose a learning-based QKD resource allocation scheme for quantum-secured FEL systems, which is strengthened by federated reinforcement learning. In particular, we use the model-free off-policy soft actor-critic (SAC)~\cite{haarnoja2018soft} structure to learn the optimal QKD resource allocation strategy. For each QKD manager and each controller, a policy network is adopted to configure the QKD nodes by learning the allocation strategy during the interaction with quantum-secured FEL systems. Moreover, a Q-network is adopted as a critic of each QKD manager and controller to evaluate the state-action values of its local policy, i.e., the performance with the local policy. To avoid direct reward sharing, QKD managers encrypt the Q-networks and then share them with QKD controllers for their local policy evaluation and improvement. In this way, the incomplete experience issues of QKD controllers can be addressed, thus improving the training efficiency of agents in the control and management layer. In addition, to further improve convergence efficiency, the local policy of QKD controllers is aggregated as the global QKD resource allocation policy for QKD managers after improving the local policies of QKD controllers.
Our contributions can be summarized as follows.
\begin{itemize}
    \item  We propose a new hierarchical architecture for quantum-secured FEL systems to resolve uncertain factors in global model aggregation while providing ITS transmission of public keys and models. This architecture is capable of protecting the transmission of FEL models from external and participant attacks.
    \item In the proposed architecture, unlike deterministic linear programming, we formulate the optimization problem as stochastic programming to resolve the uncertainty of the security demand in quantum-secured FEL systems, i.e., the required secret-key rates. Considering the dynamic factors in the global aggregation of FEL, such as the number of FEL workers, the proposed model aims to minimize the deployment cost of the systems.
    \item To solve the proposed stochastic model without complete information, we proposed a federated DRL scheme that allows QKD managers and controllers to make the optimal decision independently based only on their local partial observation. Specifically, the proposed scheme enables QKD managers and controllers to learn a global policy collaboratively while maintaining their experiences in local replay buffers. Therefore, the proposed scheme can learn a synthetic QKD resource allocation policy efficiently without prior knowledge while preserving the privacy of the learning agents.
    \item Extensive experiments demonstrate the effectiveness of stochastic and learning-based resource allocation schemes for quantum-secured FEL systems. The performance evaluation results illustrate that existing baselines in the Quantum Internet with average or random demand do not lead to acceptable solutions that are significantly inferior to the proposed schemes.
\end{itemize}

We organize the rest of this paper as follows. In Section~\ref{sec:related}, we provide a review of the related works. In Section~\ref{sec:system}, we discuss the system model. In Section~\ref{sec:optimization}, we discuss the proposed optimization solution approach. In Section~\ref{sec:algorithm}, we propose the federated deep learning-based algorithm. Finally, we conduct the simulation experiments in Section~\ref{sec:experiment} and conclude in Section~\ref{sec:conclusions}. The abbreviations and definitions used in this paper are summarized in Table \ref{table:abb}.

\section{Related Works}\label{sec:related}
\subsection{Federated Edge Learning}
Due to the enormous volume of data generated at the network edge~\cite{yang2021reconfigurable, eldar2022machine, xu2021learning}, FEL has emerged as a promising paradigm of distributed privacy-preserving learning to improve the efficiency and security of communication and sensing for edge networks~\cite{kang2022communication,li2022novel,chaccour2022seven}. To illustrate the effectiveness and efficiency of FEL, Xu \textit{et al.} in \cite{xu2022edge} provided a systematic overview for the convergence of edge networks and learning. They highlight the potential benefits of learning-based communication systems, such as semantic communications, and the necessity of sustainable resource allocation for edge learning systems.  For instance, Hardy \textit{et al.}~\cite{wang2019adaptive} proposed a two-stage federated end-to-end learning system with performance comparable to centralized learning systems, including privacy-preserving local dataset adaptation and federated logistic regression over intermediate results encrypted with additive homomorphic encryption (HE) schemes. To address the problems of multi-view sensing observations in distributed wireless sensing, Liu \textit{et al.}~\cite{liu2022vertical} proposed a vertical FEL system for cooperative detection while preserving the data privacy of sensors. By considering the non-cooperative nature of FEL participants at the edge, Lim \textit{et al.}~\cite{lim2021dynamic} proposed a hierarchical framework, including the evolutionary game and the Stackelberg game, for edge association and resource allocation problems in FEL systems.  In addition, to motivate data owners to participate in FL, Zhan \textit{et al.} in \cite{zhan2021survey} provided a comprehensive survey on how to design proper incentive mechanisms for different federated learning algorithms in heterogeneous edge networks. Considering that conventional key distribution schemes in secure FEL systems are no longer secure, Huang \textit{et al.}~\cite{huang2021starfl} proposed a new architecture for federated learning systems in the Quantum Internet called StarFL, which uses satellite and quantum key distribution schemes to distribute public keys for FEL workers with provable security.

\subsection{The Quantum Internet}

\begin{table}[t]
	\setlength{\abovecaptionskip}{5pt}
	\setlength{\belowcaptionskip}{5pt}
	\renewcommand{\arraystretch}{1.3}
	\caption{Abbreviations and definitions.}
	\label{table1}
	\centering
	\begin{tabular}{l l}
		\toprule
		Abbreviations & Definitions     \\
		\hline
        FEL           & Federated Edge Learning            \\
        QKD           & Quantum Key Distribution           \\
        AI            & Artificial Intelligence            \\
        MPC           & Multi-party Computation            \\
        TEE           & Trusted Execution Environment      \\
        HE & Homomorphic Encryption\\
        BB84          & Bennett-Brassard-1984              \\
        GG02          & Grosshans-Grangier-2002            \\
        MDI-QKD       & Measurement Device-Independent QKD \\
        ITS & Information-Theoretic Security\\
        DRL & Deep Reinforcement Learning\\
        SAC           & Soft Actor-Critic                  \\
        KM           & Key Management \\
        GKS & Global Key Server\\
        LKM & Local Key Manager\\
        QTs & Quantum Transmitters\\
        QRs & Quantum Receivers\\
        SIs & Security Infrastructures \\
        MUX/DEMUX & Multiplexing/Demultiplexing\\
        POMDP & Partially Observable Markov Decision Process\\
        SAC & Soft Actor-Critic\\
        QBN & QKD Backbone Networking\\
        EVF & Expected Value Formulation\\
        SIP & Stochastic Integer Programming\\
		
		\bottomrule
	\end{tabular}\label{table:abb}
\end{table}

The Quantum Internet~\cite{wehner2018quantum} connects quantum devices through quantum channels to provide long-term protection and future-proof security for  the transmission of confidential information. It is expected that the Quantum Internet can provide new networking technologies for numerous critical applications by fusing quantum signal transmission with the classical communication channels. Aiming to provide secure connectivity for mission-critical applications in the real world, Toudeh-Fallah \textit{et al.}~\cite{toudeh2022paving} established the first 800-Gbps quantum-secured optical channel up to 100 km, which is able to secure 258 data channels 
under AES-256-GCM with the quantum key update rate of one per second. By combining 700 fiber-based QKD links and two high-speed satellite-based QKD links, Chen \textit{et al.}~\cite{chen2021integrated} developed an integrated space-to-ground quantum communication network with total coverage of 4,600 km. However, the costly quantum devices and the non-scalable routing schemes are obstacles to the large-scale deployment of the Quantum Internet. To address the scalability issues, Mehic \textit{et al.} \cite{mehic2019novel} proposed a routing protocol for the Quantum Internet that considers geographic distance and connection state to achieve high scalability by minimizing cryptographic key consumption. Meanwhile, For 
quantum-secured communications over backbone networks, Cao \textit{et al.}~\cite{cao2019cost} developed a programming-based resource allocation model and a heuristic algorithm to address the deployment cost-minimized problem efficiently.

\subsection{Learning-based QKD Resource Allocation Schemes}

Quantum key distribution is one of the mature applications of the Quantum Internet that has already been applied in some commercial scenarios~\cite{toudeh2022paving}. However, QKD resources in the Quantum Internet still require efficient allocation schemes to bridge the gap between the low key generation rate of the Quantum Internet and the uncertain key demands of quantum-secured communication services. To fill this gap, Zuo \textit{et al.}~\cite{zuo2020reinforcement} proposed a reinforcement learning-based algorithm to learn the optimal QKD resource allocation strategy for management of the quantum key pool in a cost-effective way. To address the dynamic arrival problem in multi-tenant QKD deployment, Cao \textit{et al.}~\cite{cao2020multi} proposed a reinforcement learning-based algorithm to reduce the deployment cost of QKD resources. However, these learning-based approaches consider QKD managers and controllers can share their observations and experiences without privacy concerns, which is infeasible in privacy-preserving systems. Therefore, in this paper, we propose a privacy-preserving learning-based resource allocation to minimize the deployment cost under uncertain security requirements of FEL workers and model owners in quantum-secured FEL systems.

\section{System Model}\label{sec:system}
In this section, we first describe the proposed hierarchical architecture for quantum-secured FEL systems. We also describe the workflow in detail and give the complexity and security analysis. Based on the proposed architecture, we illustrate the network model, cost model, and uncertainty in quantum-secured FEL systems. As shown in Fig.~\ref{fig:quantum-secured FEL}, we consider a hierarchical architecture for quantum-secured FEL systems. In the FL layer, the cluster heads organize FL workers to join federated learning for model owners. The FL model encryption and transmission among cluster heads are secured by quantum cryptography. In the control and management layer, cluster heads initiate requests for quantum secret-keys to the centralized QKD manager. Upon receiving these secret-key requests, each QKD manager first queries secret-key status and then sends configurations to QKD controllers via the simple network management protocol (SNMP). According to the received instructions, QKD controllers handshake with QKD nodes for detailed configurations. In the QKD infrastructure layer, there are three types of nodes (i.e., QKD nodes, trusted relays, and untrusted relays) and two types of links (i.e., key management (KM) and QKD links). We consider that the FEL nodes are co-located with the QKD nodes and that different types of links can be multiplexed within a single fiber~\cite{cao2021hybrid}. Therefore, the topology of the QKD layer follows that of the FEL layer, which can be denoted by $G(\mathcal{V},\mathcal{E})$. Here, $\mathcal{V}$ represents the set of FEL/QKD nodes, and $\mathcal{E}$ denotes the collection of fiber connections. Within a QKD node, there is a global key server (GKS), a local key manager (LKM), and one or more quantum transmitters (QTs). Between multiple QKD nodes, a QKD chain based on mixed relays can be used for global secret-key generation, where the trusted relays contain two or twice as many QTs~\cite{cao2021hybrid}, an LKM, and security infrastructures (SIs), while the untrusted relays consist of one or more quantum receivers (QRs). For links in the QKD layer, $\sigma$ denotes the available wavelengths on QKD links and $\kappa$ those on KM links.

\begin{figure}[t]
    \centering
    \includegraphics[width=1\linewidth]{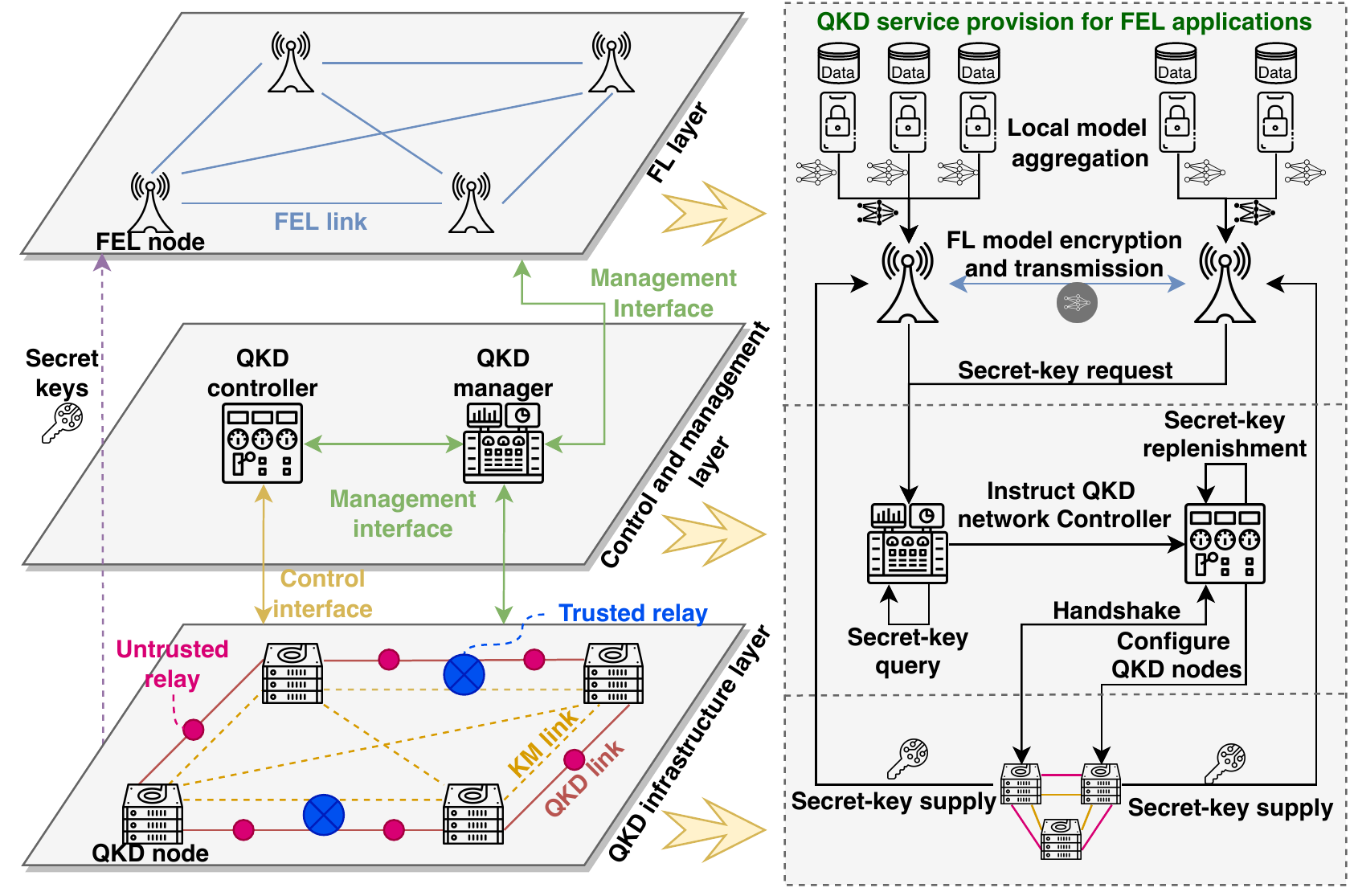}
    \caption{An overview of quantum-secured federated edge learning in the Quantum Internet}
    \label{fig:quantum-secured FEL}
\end{figure}

\subsection{Quantum-secured Federated Edge Learning}
\label{sec:FLoverQKD}
As illustrated in Fig.~\ref{fig:quantum-secured FEL}, the proposed hierarchical architecture for quantum-secured FEL systems includes three layers \cite{cao2022evolution}, i.e., the FEL layer, the control and management layer, and the QKD infrastructure layer. In the FEL layer, a set of FEL workers denoted by $\mathcal{K}=\{1,\dots,k,\dots,K\}$  at the edge networks, i.e., end devices owning training datasets, participate in FEL tasks initialized by the set of model owners $\mathcal{O}=\{1,\dots,o,\dots,O\}$ with the aid of the set of cluster heads $\mathcal{J}=\{1,\dots,j,\dots,J\}$ (e.g., base stations)~\cite{lim2021decentralized}. Suppose there is a subset of FEL workers, i.e., $K_j$ FEL workers~\cite{yang2019federated}, training a global model collaboratively with $S_o$ data samples $\{x_i,y_i\}_{i=1}^{S_o}$ for the model owner $o$. Moreover, the feature space $x_i\in\mathbb{R}^{1\times d}$ is distributed exclusively among FEL workers. The data samples in the local dataset of the FEL worker $k$ within the cluster $j$ can be denoted by $\{x_i^k \in \mathbb{R}^{1\times d_k}\}^{K_j}_{k=1}$, where $d_k$ is the dimension of its feature space. Without loss of generality, the number of FEL workers with labeled data is set to one in this paper, which is the FEL worker $K_j$. Therefore, the dataset with only features of FEL worker $k$ can be denoted as $\mathcal{D}^k_i \triangleq\{x_i^k\}$ for $k=1,\dots,K_j-1$. Meanwhile, let $\mathcal{D}_i^{K_j}\triangleq \{x_i^{K_j}, y_i^{K_j}\}$ denote the dataset of FEL worker $K_j$ with features and labels. 
Let $\theta_k\in \mathbb{R}^{d_k}$ denote the model parameters of the FEL worker $k$ and $\Theta_{j} = [\theta_1,\dots,\theta_{K_j}]$ be the union of model parameters of all the FEL workers. The goal of FEL is to minimize inference loss through collaborative training of the optimal global model $\Theta^*$ for model owners  that minimizes the loss function $\sum_{j=1}^{J} f_j(\sum_{k=1}^{K_j}\theta_k x^k, y^{K_j};\Theta_j)$.
As illustrated in Fig. \ref{fig:workflow}, the training process of quantum-secured FEL systems in the proposed hierarchical architecture consists of five main steps~\cite{yang2019federated, cheng2021secureboost, liu2019communication}:

\begin{itemize}
    \item \textbf{Step 1: Secure Communication Setup.} 
In quantum-secured FEL systems, the model owner $o$ initializes its FEL task to the FEL worker $k$ and distributes the initialized model parameters to the FEL workers via the cluster head $j$. First, the cluster head $j$ needs to establish secure communication channels with the model owner $o$ and the FEL worker $k$. To this end, the cluster head $j$ sends QKD requests to QKD managers to establish quantum-secured communication channels with the model owner $o$ and the FEL worker $k$, respectively. In detail, there are three main steps to distribute the quantum key $k^Q_{(j,k)}$ and $k^Q_{(j,o)}$ between the corresponding QKD nodes of the FEL worker $k$ and model owner $o$ via quantum channels, including transmitting the quantum bits, sifting the received bits, and estimating the error rate~\cite{bb84}. In this way, the cluster heads can establish quantum-secured communication channels using the quantum keys with model owners and FEL workers;
    
    \item \textbf{Step 2: Private Set Intersection.} In FEL systems, FEL workers have different data features and samples in their local datasets. Therefore, the model owner $o$ must find a common set of data samples $S_m$ for all participating FEL workers. The model owner can apply the \textit{cross-database intersection}~\cite{liang2004privacy} to match the data samples among FEL workers while preserving their privacy. Specifically, the cluster head $j$ applies the AES03 intersection protocol and blind signatures by generating a pair of RSA public and private keys $(e^{\textrm{RSA}}_j,d^{\textrm{RSA}}_j,n^{\textrm{RSA}})$ and using the encrypted public key $ y_{j,k}(e^{\textrm{RSA}}_j) = E_{k^Q_{j,k}}(e^{\textrm{RSA}}_j)$ to the FEL worker $k$ and $ y_{j, o}(e^{\textrm{RSA}}_j) = E_{k^Q_{j,o}}(e^{\textrm{RSA}}_j)$ to the model owner $o$ over quantum-secured communication channels. Then FEL workers and the model owner decrypt the public RSA key using their quantum keys for the intersection of the private set. By using the intersection of the private sets, the common data samples $S_o$ are discovered by the model owner $o$ without revealing the private information of FEL workers;
    
    \item \textbf{Step 3: Local Model Forward Propagation.} After determining the common data samples among FEL workers of the cluster head $j$, the FEL workers use their common local datasets $\mathcal{D}^{K_j}=\cup_{k=1}^{K_j} \mathcal{D}^k$ to train their local models. The FEL workers $k = 1,\dots,K_j-1$ with features in their training samples can only obtain intermediate results $u_k = \theta_k x^k$ from forwarding propagation and cannot compute the loss for backward propagation. Meanwhile, the FEL worker $K_j$ with labels can obtain intermediate inference results $u_{K_j}= \theta_{K_j} x^{K_j}$ and compute the local loss $f_j(\theta_{K_j} x^{K_j}, y^{K_j})$;
    
    \item \textbf{Step 4: Global Model Aggregation.} Once forward propagation is completed, the cluster head $j$ creates an HE key pair $(e^{\textrm{HE}}_j,d^{\textrm{HE}}_j)$ and distributes the public key $e^{\textrm{HE}}_j$ to FEL workers and model owners via quantum-secured channels. Let $[[\cdot]]$ denote the operation of additive HE~\cite{yang2019federated}, the FEL workers first encrypt the intermediate results as $[[u_k]]$ and/or the encrypted loss $[[f_j(\theta_{K_j} x^{K_j}, y^{K_j})]]$ using the HE public key $e^{\textrm{HE}}_j$. Then, they transmit the encrypted intermediate inference results to model owners via the cluster headers. In this way, the FEL workers can share the intermediate inference results for global gradient and loss calculation without privacy leakage;
    
    \item \textbf{Step 5: Local Model Back Propagation.} The cluster head $j$ decrypts the received intermediate results with the private key $d^{ \textrm{HE} }_j$. Then, the cluster head sends the results to the model owners and FEL workers in the quantum-secured channels. Finally, the model owner $o$ and the FEL workers decrypt and unmask the intermediate results and update their local model parameters accordingly.
\end{itemize}
\begin{figure}[t]
    \centering
    \includegraphics[width=0.8\linewidth]{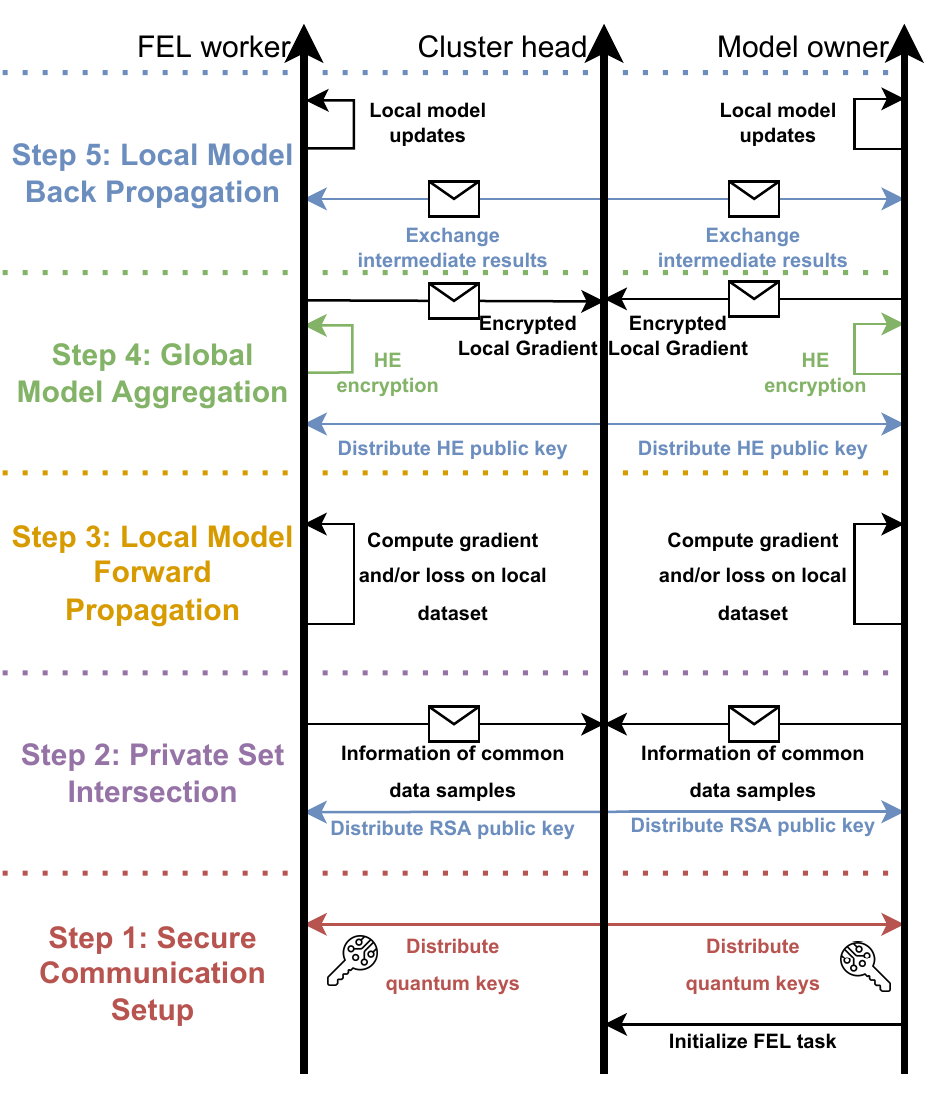}
    \caption{The workflow of quantum-secured FEL. The black lines denote the classical communication links, the red lines denote the quantum links, and the blue lines denote the quantum-secured communication links.}
    \label{fig:workflow}
\end{figure}
During the training phase and the inference phase of FEL, there are four pillars of privacy for model owners and FEL workers. First, data privacy is the biggest concern of FEL workers. Next, model privacy, including model architecture privacy and model weight privacy, is essential for model owners that must not be stolen by any malicious participants. The compromise of model privacy makes the model owner have little motivation to improve the edge learning model performance, as their trained model can be easily sniffed or stolen by other participants. On the one hand, the privacy in the input stage means that only the permitted FEL workers can input data to the model. On the other hand, privacy in the output stage means that the output of the model is only visible to the model owners. Following the literature definition in \cite{huang2021starfl}, we can have the following remarks.
\begin{remark}[Privacy-preserving Distributed Machine Learning]
    The quantum-secured FEL systems are privacy-preserving distributed ML systems with no data and model privacy leakage during the input and output stages of the training phase and the inference phase.
\end{remark}
During the input stage, FEL workers train their local models to fit their local datasets. They keep their training data local and then use the HE public key $e^{\textrm{HE}}$ received from quantum-secured communication channels to encrypt the intermediate inference results. The local gradient encrypted with the public key can only be decrypted by the model owner who holds the private key. Therefore, in the input stage, the privacy of the data and models is well-protected.

During the output stage, the cluster heads decrypt the received intermediate results by using its private key of HE $d^{\textrm{HE}}$. On the one hand, the privacy of the data is preserved by the private set overlap in step 2, since no one can interfere with the labels of FEL workers. On the other hand, the privacy of the models is enhanced since the public keys are distributed through quantum-secured communication channels so that no eavesdropper can intercept the local model updates and obtain the model architecture and model parameters from the encrypted models after encryption.


\textbf{Complexity Analysis:} Let $|k^Q|$ be the length of symmetric keys distributed via quantum channels. The complexity of generating symmetric keys for $K$ FEL workers is $\mathcal{O}(K|k^Q|)$. During the private set intersection, the complexity to generate RSA key pairs for FEL workers is $\mathcal{O}(\log^2(n^{\textrm{RSA}}))$~\cite{rivest1978method}. Moreover, the complexity of the encryption and the decryption operations is $\mathcal{O}(K\log(|e^{\textrm{RSA}}|)\log^2(n^{\textrm{RSA}})+\log^3(n^{\textrm{RSA}}))$ for K FEL workers. In FEL, the Paillier HE scheme~\cite{paillier1999public} is usually adopted for additive HE. Therefore, the algorithmic complexity of HE encryption and decryption are both $\mathcal{O}(S \log(n^{\textrm{HE}}))$.

\textbf{Security Analysis:} {Here, we can give external and participant attack examples, i.e., eavesdropping attacks and Sybil attacks, to show that the quantum-secured FEL systems can defend against common attacks in FEL.} The first one is that the eavesdropping attacks during quantum-secured FEL training cannot pose any threat to the data and models in the system~\cite{xie2021securing}. In other words, even if there is an eavesdropper in the system, the only information obtained during the eavesdropping attacks is encrypted with the unconditional secure scheme, i.e., one-time pad~\cite{shannon1949communication}. Moreover, the eavesdroppers have no way to decrypt or obtain valuable information from the ciphertext, i.e., the encrypted messages. Second, quantum-secured FEL systems are able to resist Sybil attacks~\cite{fung2018mitigating} during the training process. Each FEL worker can only participate in FEL through one single quantum-secured channel. Therefore, FEL workers cannot forge their identities and submit multiple local models.


\subsection{Networking Model for Quantum-secured FEL Systems}
Let $\iota$ denote the distance from a QT to its connected QR. Due to the symmetrical locations of two connected QTs, the distance between them can be approximated as
\begin{equation}
    L \approx 2 \cdot \iota.
\end{equation}
At the length $L$ between two connected QTs, the maximum attainable secret-key rate is denoted by $K_L$, which is in inverse proportion to the length $L$, i.e., an increment in $L$ leads to a reduction in $K_L$~\cite{cao2022evolution}.

In the proposed quantum-secured FEL model, $\mathcal{R}$ represents the set of quantum-secured FEL model transmission requests and $r(s_r,d_r,\rho_r)\in \mathcal{R}$ denotes one quantum-secured FEL model transmission request of FEL nodes, in which $s_r$ and $d_r$ represent the source node and the destination node of quantum-secured FEL model transmission request $r$, respectively. Let $\rho_r$ be the amount of concurrent quantum-secured FEL links for satisfying the security demands (i.e., secret-key rates) of the FEL training between $s_r$ and $d_r$~\cite{cao2021hybrid}, which can be calculated as
\begin{equation}
    \rho_r = \left \lceil \frac{k_r}{K_L} \right \rceil,
\end{equation}
where $k_r$ denotes the required security demand between source node $s_r$ and destination node $d_r$.



\subsection{Cost Model and Provisioning Plans}
Let $l_{(n,m)}$ denote the length of the physical link of the model transmission request between node $n$ and node $m$.
\subsubsection{QTs and QRs}

Since each MDI-QKD requires two QTs and one QR, the amount of QTs $a_{T}^r$ and QRs $a_{R}^r$ for the model transmission request $r$ of quantum-secured FEL systems can be calculated as
\begin{equation}
    a_{T}^r = \sum_{(n,m)\in E_r} 2\cdot \rho_r \cdot \left \lceil \frac{l_{(n,m)}}{L} \right \rceil,
\end{equation}
and
\begin{equation}
    a_{R}^r = \sum_{(n,m)\in E_r} \rho_r \cdot \left \lceil \frac{l_{(n,m)}}{L} \right \rceil,
\end{equation}
where $E_r$ denotes the set of physical fibers along the transmission path of request $r$.


\subsubsection{LKMs} The requested amount of LKMs $a_{KM}^r$ for a quantum-secured FEL model transmission request $r$ is
\begin{equation}
    a_{KM}^r = \sum_{(n,m)\in E_r} \left \lceil \frac{l_{(n,m)}}{L} +1 \right \rceil.
\end{equation}
\subsubsection{Security Infrastructures} 
The requested amount of SI $a_{SI}^r$ for a quantum-secured FEL model transmission request $r$ is
\begin{equation}
    a_{SI}^r = \sum_{(n,m)\in E_r} \left \lceil \frac{l_{(n,m)}}{L} - 1 \right \rceil.
\end{equation}
\subsubsection{MUX/DEMUX Components}
Finally, let $a_{M}^r$ denote the number of MUX/DEMUX component pairs for a quantum-secured FEL model transmission request $r$ is
\begin{equation}
    a_{M}^r = \sum_{(n,m)\in E_r} \left \lceil \frac{l_{(n,m)}}{L} \right \rceil + \sum_{(n,m)\in E_r} \left \lceil \frac{l_{(n,m)}}{L} -1\right \rceil,
\end{equation}
where the number of MUX/DEMUX component pairs required by QTs and QRs are represented by the first term $\sum_{(n,m)\in E_r} \left \lceil \frac{l_{(n,m)}}{L} \right \rceil$ and the second term $\sum_{(n,m)\in E_r} \left \lceil \frac{l_{(n,m)}}{L} -1\right \rceil$, respectively. 

\subsubsection{QKD and KM Links}
In the case of this study, three wavelengths are occupied by one QKD link and three wavelengths by one KM link~\cite{wonfor2019field}. For a quantum-secured FEL model transmission request $r$, the link cost can be calculated as
\begin{equation}
    a_{Ch}^r = \sum_{(n,m)\in E_r} (3 \rho_r l_{(n,m)} + l_{(n,m)}),
\end{equation}
where the required length of QKD and KM links and FEL links can be represented by $3\rho_r l_{(n,m)}$ and $l_{(n,m)}$, respectively.
\subsubsection{Provisioning Plans and Deployment Cost}

When allocating resources for quantum-secured FEL applications in the proposed system model, the QKD manager needs to consider either a reservation plan or an on-demand plan, which is similar to the cloud and other online services~\cite{chaisiri2011optimization,chen2019prediction}. The reservation plan is for long-term allocation/subscription of QKD resources (specify what those are), while the on-demand plan is for short-term demand. If the QKD manager knows the demand of each FEL node in the FEL layer, it can provide the optimal reservation plan in the QKD layer. However, the demand of each FEL node is random as the number of workers in the FEL layers is uncertain. For example, some FEL nodes may require different secret-key rates in model transmission requests to protect their workers' models. According to the above two types of subscription plans, each available QKD resource has two corresponding subscription costs, namely reservation costs and on-demand costs. We define the cost function as monetary units (e.g., dollars) per unit of QKD resources. In the reservation phase, $\beta_{T}^b, \beta_{R}^b, \beta_{ KM }^b, \beta_{ SI }^b,\beta_{ MD }^b$, and $\beta_{Ch}^b$, are the reservation costs for the QTs, QRs, LKMs, SIs, MUX/DEMUX components, and QKD and KM links, respectively. 
Also, $\beta_{T}^o, \beta_{R}^o, \beta_{ KM }^o, \beta_{ SI }^o,\beta_{ MD }^o$, and $\beta_{Ch}^o$, are the on-demand cost for the QTs, QRs, LKMs, SIs, MUX /DEMUX components, and QKD/KM links, respectively.


\subsection{Uncertainty of Security Demands in Quantum-secured FEL Systems}

With uncertainty of requests, the amount of required QKD resources by the FEL workers is not precisely known when the resources are reserved. In quantum-secured FEL, the encryption of public keys and intermediate inference results consume a certain amount of key resources. However, cluster heads choose different numbers of FEL workers to meet model owners' requirements for model accuracy and training error. Unfortunately, the model accuracy and training error are affected by various dimensional parameters such as data volume, algorithm quality, and FEL tasks. Therefore, cluster heads cannot accurately predict the model accuracy and thus require uncertain secret-key rates in the training process of FEL models. Let $\mathcal{K}_r = \{0,1,\ldots,K\}$ represent the set of possible secret-key rates requirements of request $r\in \mathcal{R}$. The set of all possible secret-key rates of QKD nodes $\mathcal{K}$ in the QKD layer can be represented by the Cartesian product as
\begin{equation}
    \mathcal{K} = \prod_{r\in R} \mathcal{K}_r = \mathcal{K}_1 \times \mathcal{K}_2 \times \cdots \times \mathcal{K}_{|\mathcal{R}|}.
\end{equation}
The probability distributions for both secret-key rates in $\mathcal{K}$ of all secure model transmission requests $\mathcal{R}$ are considered to be known. Note that statistical processes can be used to analyze historical data and that ML methods can predict the distribution of these demands.

\section{The Proposed Optimization Models}\label{sec:optimization}
In this section, we first formulate the QKD resource allocation problem as a deterministic linear programming model. Furthermore, considering the uncertainty, i.e., the required secret-key rates, in quantum-secured FEL systems, we develop a stochastic programming model for QKD resource allocation. In the resource allocation model, the management signals are usually brief and can be encoded as quantum information and transmitted in QKD links, thus the data transmission is secured by quantum cryptography.


\subsection{Deterministic Integer Programming}

Initially, consider the case where the actual FEL security requirements of FEL nodes are precisely known, and QKD resources can be subscribed in a reservation plan, where each link has two decision variables.
\begin{enumerate}
    \item $X = \{X_{(n,m)}|(n,m)\in \mathcal{E}\}$ indicates the set of numbers of wavelengths allocated in each KM link, e.g., $X_{(n,m)} = 1$ means that the QKD manager reserves one wavelength for the a model transmission request.
    \item $F = \{F_{(n,m)}|(n,m)\in \mathcal{E}\}$ indicates the set of the numbers of wavelengths allocated in each QKD link.
\end{enumerate}

Based on the cost model and demand, the deterministic integer programming for the reservation plan can be formulated to minimize the total cost of the QKD problem as follows:
\begin{equation}
\begin{aligned}
    \min_{X^b, F^b} &C^b(X^b,F^b)
    =\sum_{(n,m)\in \mathcal{E}} \bigg[ \frac{F_{(n,m)}^b}{3}(a_{T}^b\beta_{T}^b  + a_{R}^b\beta_{R}^b)  \\&+X_{(n,m)}^b(a_{KM}^b\beta_{KM}^b  + a_{SI}^b\beta_{SI}^b + a_{MD}^b\beta_{MD}^b) \\& +  l_{i,j} (F_{(n,m)}^b + X_{(n,m)}^b)\beta_{Ch}^b\bigg],
\end{aligned}\label{eq:10}
\end{equation}
subject to:
\begin{alignat}{2}
\begin{split}
    &\sum_{j\in \mathcal{V}}\sum_{p\in \kappa} x_{(n,m),p}^r - \sum_{j\in \mathcal{V}}\sum_{p\in \kappa} x_{(m,n),p}^r \\& = \begin{cases}
  1& \text{ if } i=s_r \\
  -1& \text{ if } i=d_r \\
  0&  \text{ otherwise }
\end{cases}, \forall r\in R,\forall (n,m)\in \mathcal{E} \label{eq:11}
\end{split}\\
&X_{(n,m)}^b = \sum_{r\in R}\sum_{p\in \kappa} x^r_{(n,m),p}, \forall (n,m)\in \mathcal{E} \label{eq:12}\\
&F_{(n,m)}^b = \sum_{r\in R}\sum_{q\in \sigma} f_{(n,m), q}^r, \forall (n,m)\in \mathcal{E} \label{eq:13}\\
&\sum_{q\in \sigma} f_{(n,m),q}^r = 3 \rho_r \sum_{p\in \kappa} x_{(n,m),p}^r, \forall r\in R, (n,m)\in \mathcal{E} \label{eq:14}\\
&\sum_{j\in \mathcal{V}} f_{(n,m),q}^r  = \sum_{j\in \mathcal{V}} f_{(j,i),q}^r, \forall r\in R, i \in \mathcal{V}, q \in \sigma \label{eq:15}\\
&\sum_{j\in \mathcal{V}} x_{(n,m),p}^r  = \sum_{j\in \mathcal{V}} x_{(j,i),p}^r, \forall r\in R, i \in \mathcal{V}, p \in \kappa \label{eq:16}\\
&\sum_{r\in R}\sum_{p\in \kappa} x_{(n,m),p}^r \leq |\kappa|,\quad \forall (n,m)\in \mathcal{E} \label{eq:17}\\
&\sum_{r\in R}\sum_{q \in \sigma} f_{(n,m),q}^r \leq |\sigma|,\forall (n,m)\in \mathcal{E} \label{eq:18}\\
&\sum_{r\in R} x^r_{(n,m),p} \leq 1, \quad \forall (n,m)\in \mathcal{E}, p \in \kappa \label{eq:19}\\
&\sum_{r\in R} f^r_{(n,m),q} \leq 1, \quad \forall (n,m)\in \mathcal{E}, q \in \sigma \label{eq:20}
\end{alignat}
The flow conservation constraint Eq. \eqref{eq:11} guarantees that the QKD path flows between two distant FEL nodes is one in source/destination QKD nodes and zero in other QKD nodes. That is due to specifying a QKD path for a quantum-secured FEL request with a dedicated source and destination QKD nodes. The constraints in Eq. \eqref{eq:12} and Eq. \eqref{eq:13} guarantee that the demand is satisfied where $\sum_{p\in \kappa} x_{(n,m), p}^r$ and $\sum_{q\in \sigma} f_{(n,m),q}$ are the requested numbers of KM and QKD wavelengths in link $(n,m)$, respectively. The constraint in Eq. \eqref{eq:14} specifies the number of wavelength channels requested by the QKD and the KM links of each quantum-secured FEL model transmission request, where the QKD/KM links and the FEL links are required to take up $3\rho_r$ and one wavelength channels, respectively. The same wavelength channel constraints in Eq. \eqref{eq:15} and Eq. \eqref{eq:16} are the constraints on wavelength continuity, which ensures that the identical wavelength channel is allocated to the link on the chosen path of each secure model transmission request. The wavelength capacity constraints Eq. \eqref{eq:17} and Eq. \eqref{eq:18} guarantee that the total wavelength channels of the QKD/KM links should not exceed the available wavelength channels. The wavelength uniqueness constraints Eq. \eqref{eq:19} and Eq. \eqref{eq:20} guarantee either zero or one wavelength channel can be allocated for the secure model transmission requests.

\subsection{Stochastic Integer Programming}

The deterministic integer programming developed in Eqs. \eqref{eq:10}--\eqref{eq:20} is no longer applicable if the demand for the resources is unknown. Therefore, we describe the stochastic integer programming (SIP), which optimizes the overhead of the QKD resources allocated to all quantum-secured FEL model transmission requests. The first phase includes all reservations that must be determined before the requirements can be implemented and analyzed. It is critical for the QKD manager to reserve the number of QKD resources to be utilized before the demand is observed. In the second phase, allocations are made to accommodate real-time demand. After observing the real-time demand, if the reserved QKD link resources are less than the demand, FEL nodes have to pay for the cost of the additional QKD resources required.

The variables $X^o = \{X_{(n,m)}^{o}|(n,m)\in \mathcal{E}\}$ and $ F^o = \{F_{(n,m)}^{o}|(n,m)\in \mathcal{E}\}$ denote the sets of number of requests served in KM links and QKD links in the second stage, respectively. The expected overhead in the second stage is formulated as function $\mathbb{E}_\Omega[L(X^{o}, F^{o},\omega)]$, where $\omega \in \Omega = \mathcal{K}$ denotes the set of possible secret-key rates (called realizations, in general) observed in the second stage.

Thus, the total objective of this SIP model under uncertainty~\cite{chaisiri2011optimization} is
\begin{equation}
\begin{aligned}
    \min_{X^{b}, F^{b}, X^{o}, F^{o}} C (X^{b}, F^b)+ \mathbb{E}_\Omega[L(X^{b}, F^{b},\omega)],
\end{aligned}\label{eq:21}
\end{equation}
where
\begin{equation}
    L(X^{b}, F^{b},\omega) = \min_{Y=\{X^{o}(\omega), F^{o}(\omega)\}} L(Y) \label{eq:22}
\end{equation}
is the cost function in the second phase for given realization $\omega$. The deterministic equivalent SIP of Eqs. \eqref{eq:10}--\eqref{eq:20} for QKD resource allocation is expressed as Eqs. \eqref{eq:23}--\eqref{eq:33}.
In the optimization objective Eq. \eqref{eq:23}, there are probabilities $p(k)$, each denoting the probability of demand $k\in\mathcal{K}$ being realized. Eq. \eqref{eq:24} is the flow conservation constraint. Eq. \eqref{eq:25} and Eq. \eqref{eq:26} are demand satisfaction constraints. Eq. \eqref{eq:27} is the wavelength channel number constraints. Eq. \eqref{eq:28} and Eq. \eqref{eq:29} are the constraints to ensure each secure model transmission requests the same wavelength channel. Eq. \eqref{eq:30} and Eq. \eqref{eq:31} are wavelength capacity constraints, while Eq. \eqref{eq:32} and Eq. \eqref{eq:33} are wavelength uniqueness constraints.
\begin{figure*}[t]
\vspace{-0.2cm}
\begin{alignat}{2}
    \min_{X^{b}, F^{b}, X^{o}(k), F^{o}(k)} \quad &  C^b (X^{b}, F^b) + \sum_{k \in \mathcal{K}} p(k) \left[C^{o} \left(X^{o}(k), F^{o}(k)\right)\right] \label{eq:23}\\
	\textrm{s.t.}\quad & \sum_{j\in \mathcal{V}}\sum_{p\in \kappa} x_{(n,m),p}^r(k) - \sum_{j\in \mathcal{V}}\sum_{p\in \kappa} x_{(j,i),p}^r(k) = \begin{cases}
  1& \text{ if } i=s_r \\
  -1& \text{ if } i=d_r \\
  0&  \text{ otherwise }
\end{cases}, \quad \forall r\in R,  k\in \mathcal{K} \label{eq:24}\\
&X^{b}_{(n,m)} + X^{o}_{(n,m)}(k) \geq \sum_{r\in R}\sum_{p\in \kappa} x_{(n,m), p}^r(k), \quad \forall (n,m)\in \mathcal{E}, k\in \mathcal{K} \label{eq:25}\\
&F^{b}_{(n,m)} + F^{o}_{(n,m)}(k) \geq \sum_{r\in R} \sum_{q\in \sigma} f_{(n,m), q}^r(k), \quad \forall (n,m)\in \mathcal{E}, k\in \mathcal{K} \label{eq:26}\\
&\sum_{q\in \sigma} f_{(n,m),q}^r(k) = 3 \rho_r(k) \sum_{p\in \kappa} x_{(n,m),p}^r(k), \quad \forall r\in R, (n,m)\in \mathcal{E}, k\in \mathcal{K} \label{eq:27}\\
&\sum_{j\in \mathcal{V}} f_{(n,m),q}^r(k)  = \sum_{j\in \mathcal{V}} f_{(j,i),q}^r(k), \quad \forall r\in R, i \in \mathcal{V}, q \in \sigma, k\in \mathcal{K} \label{eq:28}\\
&\sum_{j\in \mathcal{V}} x_{(n,m),m}^r(k)  = \sum_{j\in \mathcal{V}} x_{(j,i),p}^r(k), \quad \forall r\in R, i \in \mathcal{V}, p \in \kappa, k\in \mathcal{K} \label{eq:29}\\
&\sum_{r\in R}\sum_{p\in \kappa} x_{(n,m),p}^r(k) \leq |\kappa|,\quad \forall (n,m)\in \mathcal{E}, k\in \mathcal{K} \label{eq:30}\\
&\sum_{r\in R}\sum_{q \in \sigma} f_{(n,m),q}^r(k) \leq |\sigma|,\quad \forall (n,m)\in \mathcal{E}, k\in \mathcal{K} \label{eq:31}\\
&\sum_{r\in R} f^r_{(n,m),q}(k) \leq 1, \quad \forall (n,m)\in \mathcal{E}, q \in \sigma, k\in \mathcal{K} \label{eq:32}\\
&\sum_{r\in R} x^r_{(n,m),p}(k) \leq 1, \quad \forall (n,m)\in \mathcal{E}, p \in \kappa, k\in \mathcal{K} \label{eq:33}
\end{alignat}
\hrulefill
\vspace{-0.2cm}
\end{figure*}
\section{The Proposed Learning-based Scheme}\label{sec:algorithm}

In this section, we formulate the above QKD resource allocation problem as a learning task. In particular, we model the stochastic model as a partially observable Markov decision process (POMDP) for QKD managers and controllers. Based on the properties of the POMDP, we design a learning-based QKD resource allocation scheme based on federated reinforcement learning~\cite{xu2021multi,wang2020attention} to explore the optimal solution without sharing the privacy experiences. Moreover, we give the convergence analysis and the complexity analysis of the proposed scheme.

\subsection{POMDP for quantum-secured FEL systems}

\subsubsection{State Space} The state space $s(t) \in \mathcal{S}$ of the QKD networks can be represented by the reservation strategy and the on-demand deployment of the QKD manager and QKD controllers at $t$th slot. Therefore, the state of QKD networks is defined as $s(t) := [X^b(t-l), F^b(t-l),\dots, X^b(t-1), F^b(t-1)]$. This indicates that the state of the QKD network is composed of the past $l$ reservation strategies and on-demand deployment. Therefore, for each QKD controller, their observation is a part of the global state, which can be denoted as $s_n(t) := [X^b_n(t-l), F^b_n(t-l),\dots , X^b_n(t-1), F^b_n(t-1)]$, $n = 1,\dots,N$. Without loss of generality, we consider that one QKD manager is co-located with one QKD controller but only some of the QKD controllers are QKD managers. Therefore, the observation of QKD manager $m$ can be denoted as $s_m(t) := [X^b_m(t-l), F^b_m(t-l),\dots, X^b_m(t-1), F^b_m(t-1)]$, for $m = 1,\dots,M$.
\subsubsection{Action Space}
The action space $a(t)=[a_1(t),\dots,a_N(t)]\in\mathcal{A}$ is the reservation strategy profile of QKD resource allocation deployment at time slot $t$, i.e., $a_n(t) = [X^b_n(t),F^b_n(t)], n=1, \dots, N$.
\subsubsection{State Transition Probability Function}
The state transition probability function $\mathcal{P}:\mathcal{S}\times\mathcal{A}\times\mathcal{S}\rightarrow [0,1]$ represents the changing rules of uncertain factors. The training environment transitions to the next moment of state based on the current state of the environment and the action input by the training agent, calculated by the probability function for the state transition. Therefore, we can assume $s(t+1) = s'$ and the probability of transitioning to this state is $P(s'|s,p)$, where $s(t)=s$ and $p(t)=p$. The agent's policy can change the trajectory of the state transition.
\subsubsection{Reward}
In the deployment cost minimization problem defined in Eq. (\ref{eq:23}), the reward in time slot $t$ is the ratio of near-minimal deployment cost to current deployment cost. However, the reward is only observable by the QKD manager in the Quantum Internet. Therefore, the reward function is always zero for all QKD controllers, i.e., $R_n(s(t),a(t))=0, n=1,\dots, N$. Moreover, the reward function can be expressed as the normalized cost
$R_m( s(t), a(t)) = \frac{C^b (X^{b}, F^b) + \sum_{k \in \mathcal{K}} p(k) \left[C^{o} \left(X^{o}(k), F^{o}(k)\right)\right]}{C^b (X^{b,*}, F^{b,*}) + \sum_{k \in \mathcal{K}} p(k) \left[C^{o,*} \left(X^{o,*}(k), F^{o,*}(k)\right)\right]}$ for the QKD manager $m$.
\subsubsection{Q-function} Let $\gamma$ be the discounting factor that determines how much QKD managers care about future cumulative returns versus immediate rewards. Therefore, the Bellman equation for $Q^\pi$ w.r.t the policy $\pi$ is
\begin{equation}
    \begin{aligned}
        Q^{\pi}(s(t),a(t)) = &\underset{s(t+1),a(t+1) \sim \mathcal{P^\pi}}{\mathbb{E}}\bigg[R\big(s(t),a(t)\big) \\& + \gamma(Q^{\pi}(s(t+1),a(t+1)) \\&+ \tau H\big(\pi(\cdot|s(t+1)) \big)\bigg],
    \end{aligned}
\end{equation}
where $H(\pi(\cdot|s(t+1))) = - \log \pi(\cdot|s(t+1))$ is the expected entropy for policy $\pi$ in state $s(t+1)$. To determine the relative importance of the entropy term, $\tau$ is the additive temperature parameter that help affects the optimal policy to determine its stochasticity. Therefore, the Q-function in the SAC algorithm~\cite{haarnoja2018soft} can be expressed as
\begin{equation}
    \begin{aligned}
        Q^{\pi}(s(t),a(t)) = &\underset{s(t+1),a(t+1) \sim \mathcal{P^\pi}}{\mathbb{E}}\bigg[R\big(s(t),a(t)\big) \\& + \gamma(Q^{\pi}(s(t+1),a(t+1)) \\&- \tau \log\big(\pi(\cdot|s(t+1)) \big)\bigg].
    \end{aligned}
\end{equation}

\subsection{Federated DRL-based QKD Allocation Scheme}

To help each QKD manager and controller evaluate its QKD resource allocation function $\pi(\cdot)$ parameterized by $\varphi$, a Q-function $Q(\cdot,\cdot;\vartheta)$ parameterized by a neural network $\vartheta$ is developed. After fitting the Q-network to the allocation policy, the policy network can be trained using the Q-network via the policy gradient algorithm. The Q-network can specify the expected returns for the actions performed by the policy network. In this way, the probability of actions leading to lower expected QKD deployment costs is increased. The probability of actions leading to higher expected costs decreases until the policy converges to the Q-network. The learning agent needs to repeat this process until it converges to the optimal policy.
At each local iteration of the learning agents, training experiences are sampled randomly from the replay buffers of QKD managers or QKD controllers to update the network parameters of policy or Q-networks, respectively.

\subsubsection{Local Policy Iteration for QKD Managers}
The parameters $\vartheta_{m,i}, i=1,2$ in double Q-networks of the QKD manager $m$ is updated by minimizing the difference between the output of Q-networks and the target is performed via gradient descent as follows:
\begin{equation}
    \begin{aligned}
        \epsilon_{m,i}^{e+1} = \arg\min_{\vartheta_{m,i}}\frac{1}{|B_m|} \sum_{(s(t),a(t))\sim B_m} &\bigg[Q(s(t),a(t);\vartheta_{m,i})\\&-y_m(s(t),a(t))\bigg]^2,
    \end{aligned}\label{eq:mp}
\end{equation}
where the target $y_m(s(t),a(t))$ is given by
\begin{equation}
    \begin{aligned}
        y_m(s(t),a(t)) = &R_m\big(s(t),a(t)\big) \\&+ \gamma\bigg[\min_{i=1,2}Q(s(t+1),\tilde{a}(t+1);\vartheta_{m,i,targ}) \\&-\tau\log \pi(\tilde{a}(t+1)|s(t+1);\varphi_m)\bigg],
    \end{aligned}
\end{equation}
for $i=1,2$ and $\tilde{a}(t+1) \sim\pi(\cdot|s(t+1);\vartheta_m)$.

Moreover, the policy of the learning agent in the proposed scheme should, in each state, maximize the expected future return plus expected future entropy. To this end, by utilizing the min-double-Q trick~\cite{hasselt2010double} and the entropy loss term, the policy network of QKD manager $m$ is updated to maximize the trained Q-network as follows:
\begin{equation}
\begin{aligned}
    \varphi^{e+1}_{m} = & \arg\max_{\varphi_{m}}\frac{1}{|B_{m}|} \sum_{s\in B_{m}}\bigg[\min_{i=1,2} Q(s,\tilde{a}(s;\varphi_{m});\vartheta_{m,i}) \\& - \tau \log(\pi(\tilde{a}(s;\varphi_{m})|s;\varphi_{m})) \bigg],
\end{aligned}\label{eq:mq}
\end{equation}
where $\tilde{a}$ is reparameterized from a squashed Gaussian policy w.r.t. the mean $\epsilon(s;\varphi_{m})$ and variance $\sigma(s;\varphi_{m})$ output from the policy network $\varphi_{m}$ as follows:
\begin{equation}
    \tilde{a}(s,\xi;\varphi_{m});\varphi_{m}) = \tanh\left( \epsilon(s;\varphi_{m}) + \sigma(s;\varphi_{m}) \odot \xi \right),
\end{equation}
and $\xi$ is an input noise vector sampled from some fixed distribution.

\begin{figure}[t]
\centering
\includegraphics[width=1\linewidth]{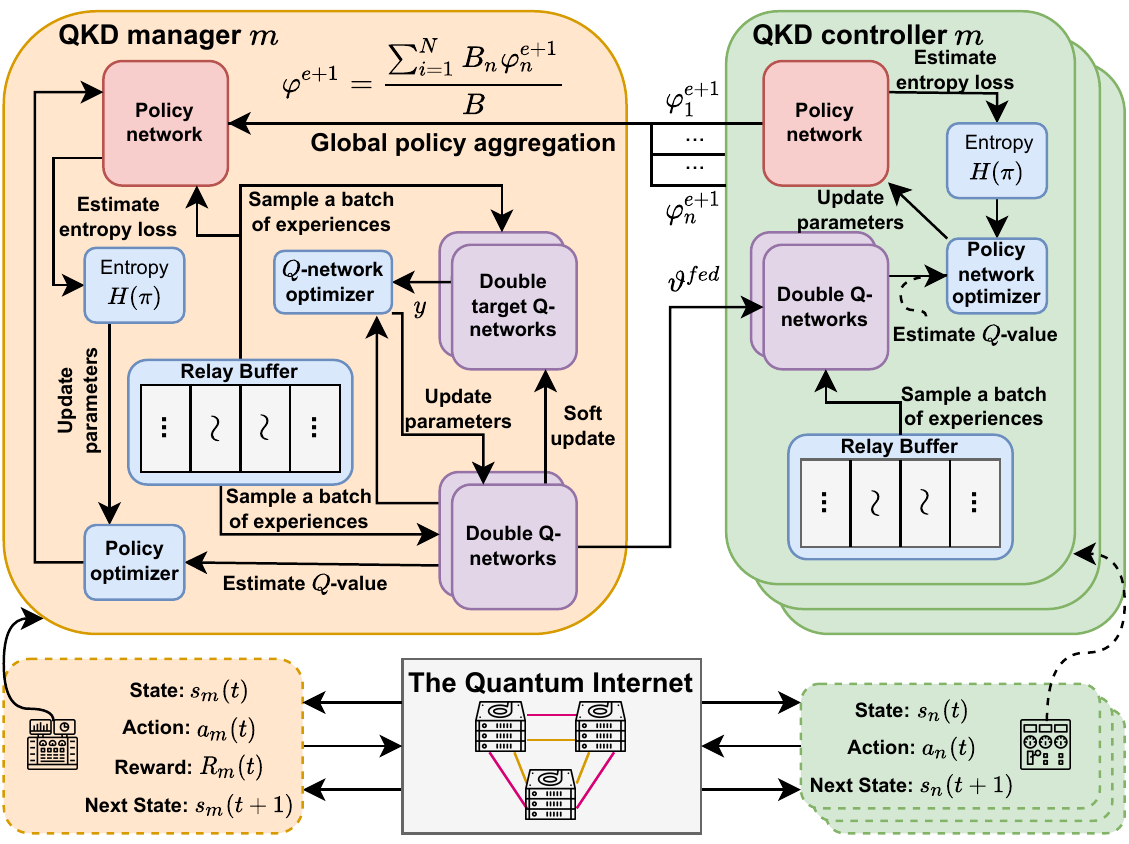}
\caption{The proposed learning-based resource allocation scheme for the Quantum Internet.}
\label{fig:algorithm}
\end{figure}

Finally, the target Q-networks are updated with $\varphi^{e+1}_{m,i,targ} \leftarrow \zeta \vartheta_{m,i,targ}^{e} + (1-\zeta)\vartheta_{m,i}^{e+1}$, for $i = 1,2$.
\subsubsection{Local Policy Iteration for QKD Controllers}
Since there is no replay buffer in the QKD controller that can be trained with Q-network, QKD controllers cannot train their policy networks independently. Fortunately, QKD controllers can borrow the trained Q-network from QKD managers to facilitate the update of its local policy, i.e. $\vartheta_{i}^{e, fed}\leftarrow \vartheta_{m,i}^e$, for $i=1,2$.
We then update the policy network of QKD controller $n$  by maximizing the objective via gradient ascent as follows:
\begin{equation}
\begin{aligned}
    \varphi^{e+1}_{n} = & \arg\max_{\varphi_{n}}\frac{1}{|B_{n}|} \sum_{s\in B_{n}}\bigg[\min_{i=1,2} Q(s,\tilde{a}(s;\varphi_{n});\vartheta_{i}^{e, fed}) \\& - \tau \log(\pi(\tilde{a}(s;\varphi_{n})|s;\varphi_{n})) \bigg],
\end{aligned}\label{eq:cp}
\end{equation}
where $\tilde{a}(s;\varphi_{n})$ is a sample from $\pi(\cdot|s;\varphi^{e}_{n})$ which is differentiable w.r.t. $\varphi^{e}_{n}$ via the reparametrization trick~\cite{haarnoja2018soft}.

\subsubsection{Global Policy Aggregation}

	Without exposing the private information, the global policy is aggregated from the local QKD allocation policy while maintaining their private experiences in their local replay buffers. The global aggregation process via FedAvg~\cite{mcmahan2017communication} of local policy networks can be presented as:
\begin{equation}
\varphi^{e+1} = \frac{\sum_{i=1}^{N}B_n\varphi_n^{e+1}}{B},\label{eq:agg}
\end{equation}
where $B = \sum_{n=1}^{N}B_n$ is the number of total incomplete experiences in the replay buffers of QKD controllers. And the updated global policy $\varphi^{e+1}$ is used for next-episode training. Algorithm 1 summarizes the proposed intelligent resource allocation scheme based on federated DRL for quantum-secured FEL systems.

\subsection{Convergence Analysis and Complexity Analysis}

Denote $\varphi^*$ as the corresponding critic under an optimal QKD resource allocation policy. We have the following assumptions\cite{wang2020federated}:
	\begin{assumption}
		\label{as1}
		For all $C_n(\varphi_n)$,
		\begin{itemize}
			\item $C_n(\varphi_n)$ is convex;
			\item $C_n(\varphi_n)$ is $\varepsilon$-smooth, i.e., $C_n(\varphi_n')\leq C_n(\varphi_n)+\nabla C_n(\varphi_n)\cdot(\varphi_n'-\varphi_n)+\frac{\varepsilon}{2}\left \| \varphi_n - \varphi_n'\right \|^2$, for $\forall \varphi_{n}$ and $\varphi_{n}'$.
		\end{itemize}
	\end{assumption}

	\begin{algorithm}[t]
		\caption{The proposed intelligent resource allocation scheme based on federated DRL}
		\label{alg1}
		{\small
			Initialize $\varphi_m, \vartheta_{m,1}, \vartheta_{m,2}, \vartheta_{1}^{fed}, \vartheta_{2}^{fed}, B_n, B_m$.\\
			\For{Episode $e \in 1, \ldots, E$}
			{
			\For{each decision slot $t$}{
			    The QKD manager $m$ and QKD controller $n$ observe $s_m(t)$ and $s_n(t)$, respectively.\\
			    Input $s_m(t)$ and $s_n(t)$ to policy network and obtain the reservation strategy profile $a_n(t)$ of QKD resource allocation.\\
			    The QKD manager $m$ store $(s_m(t), a_m(t), R_m, s_m(t+1))$ in $B_m$ and the QKD controller $n$ store $(s_n(t), a_n(t), s_n(t+1))$ in $B_n$.\\
			}
			\For{each local policy iteration}{
			The QKD managers update the policy networks according to Eq. (\ref{eq:mp}) and the Q networks according to Eq. (\ref{eq:mq}).\\
			The QKD managers updates target Q-networks with $\varphi^{e+1}_{m,i,targ} \leftarrow \zeta \vartheta_{m,i,targ}^{e} + (1-\zeta)\vartheta_{m,i}^{e+1}$, for $i = 1,2$.\\
			The QKD controllers borrow the trained Q-network from QKD managers to facilitate the update of its local policy, i.e. $\vartheta_{i}^{e, fed}\leftarrow \vartheta_{m,i}^e$, for $i=1,2$.\\
			The QKD controllers update the local policy networks according to Eq. (\ref{eq:cp}).\\
			The local policy networks of QKD controllers are aggregated in Eq. (\ref{eq:agg}) and then sent to QKD managers.
			}
			}
		}\label{alg}
	\end{algorithm}
	
	Assumption \ref{as1} is used to provide the feasibility of solution space and guarantee the update rule of federated reinforcement learning-based scheme. Then, we can have the lemma as follows:
	
	\begin{lemma}
		$C(\varphi) = \sum_{n=1}^{N}B_n C_n(\varphi_n)/B$ is $\epsilon$-strongly convex and $\varepsilon$-smooth.
	\end{lemma}
	
	\begin{proof}
	According to Assumption \ref{as1}, the global cost function $C(\varphi)$ the a weighted finite-sum of local cost function $C_n(\varphi_n)$. Furthermore, by applying the triangle inequality and the definition of convex, $C(\varphi)$ is $\epsilon$-strongly convex and $\varepsilon$-smooth.
	\end{proof}
	
	\begin{theorem}
		Considering that $C(\varphi)$ is $\varepsilon$-smooth and $\epsilon$-strongly, let $\nu = 1/C$ and $\varphi^*=\arg\min_\varphi C(\varphi)$, we have
		\begin{equation}
		\left \| \varphi^e - \varphi^* \right \| \leq (1-\frac{\epsilon}{\varepsilon})^e\left \| \varphi^1 - \varphi^*\right \|,
		\end{equation}
		so the gradient dispersion can be derived as $O(\bar{\varphi}) = \frac{\varepsilon}{\epsilon}\log(\left\|\varphi^1 - \varphi^*\right\|/\bar{\varphi})$, which is used to illustrate how the parameters $\varphi^t$ are distributed in each worker.
	\end{theorem}
	
	\begin{proof}
		According to the $\epsilon$-strongly convexity of $C(\varphi)$, we have
		\begin{equation}
		\nabla C(\varphi)(\varphi-\varphi^*) \geq C(\varphi)-C(\varphi^*)+\frac{\epsilon}{2}\left\|\varphi-\varphi^*\right\|^2.
		\end{equation}
		Thus, we can obtain the following:
		\begin{equation}
		\begin{aligned}
		&\left\|\varphi^{e+1}-\varphi^*\right\|^2 = \left\|\varphi^e -\nu\nabla C(\varphi^e)-\varphi^* \right\|^2\\
		& = \left\| \varphi^e - \varphi^*\right\|-2\nu \nabla C(\varphi^e)(\varphi^e-\varphi^*)+\nu^2\left\| \nabla C(\varphi^e)\right\|^2\\
		& \leq \left\| \varphi^e - \varphi^* \right\| - 2\nu(C(\varphi)-C(\varphi^*)\\&+\frac{\epsilon}{2}\left\|\varphi^e-\varphi^*\right\|^2)+\nu^2\left\| \nabla C(\varphi^e)\right\|.
		\end{aligned}\label{eq:45}
		\end{equation}
		By smoothing $C(\varphi)$, the gradient bound can be obtained as
		\begin{equation}
		\begin{aligned}
		C(\varphi^*) & \leq C(\varphi - \frac{1}{\varepsilon}\nabla C(\varphi))\\
		& \leq C(\varphi) - \left\| \nabla C(\varphi) \right\|^2 + \frac{1}{2\varepsilon}\left\| \nabla C(\varphi) \right\|^2\\
		& \leq C(\varphi) - \frac{1}{2\varepsilon}.
		\end{aligned}\label{eq:46}
		\end{equation}
		Combining~\eqref{eq:45}, \eqref{eq:46} can be reformulated as
		\begin{equation}
		\begin{aligned}
		&\left\| \varphi^{e+1} - \varphi^* \right\| = \left\| \varphi^{e} - \nu \nabla C(\varphi^e) - \varphi^* \right\|^2\\
		&\leq \left\| \varphi^{e} - \varphi^* \right\| - \nu\epsilon\left\| \varphi^{e+1} - \varphi^* \right\|+2\nu(\nu\varepsilon-1)(C(\varphi-C(\varphi^*)))\\
		& \leq (1-\frac{\epsilon}{\varepsilon})\left\| \varphi^{e} - \varphi^* \right\| \leq (1-\frac{\epsilon}{\varepsilon})\left\| \Delta^*(\varphi) \right\|,
		\end{aligned}
		\end{equation}
		where $\nu$ is set as the last iteration.
	\end{proof}
	
	The expected convergence bound can be estimated as
	\begin{equation}
	[C(\varphi^e)-C(\varphi^*)] \leq \bar{\varphi}^e[\Delta^e(C(\varphi^*))],
	\end{equation}
	where $C(\varphi)$ is proven to be bounded as $\Delta^e(C(\varphi^*)) = C(\varphi^1)-C(\varphi^*)$. Our proposed learning-based scheme can find an optimal QKD resource allocation strategy without sharing information and experience. As for the complexity of the proposed learning-based QKD allocation algorithm, each QKD manager and controller maintains its local policy and make independent decisions independently during the QKD resource allocation phase. Moreover, the input and output dimensions are constant and determined by the dimensions of the observation and action spaces. Therefore, the computation complexity is $O((N+M)UL)$ in each decision slot, where $L$ is the number of hidden layers and $U$ is the number of hidden neurons in each hidden layer.

\section{Experimental Evaluation}\label{sec:experiment}
This section evaluates the effectiveness of the proposed SIP scheme and the learning-based QKD resource allocation scheme in different system settings and topologies. First, we describe the setting of the experiments. Then, we analyze the cost structure of the proposed SIP scheme. Moreover, we compare the proposed scheme with other baseline schemes to demonstrate its effectiveness in resolving uncertainty in QKD resource allocation. Finally, a convergence analysis is performed to illustrate the performance of the proposed learning-based scheme.

\subsection{Experiment Settings}
Similar to~\cite{cao2021hybrid}, experiments are performed on two well-known topologies (i.e., the 14-node NSFNET topology and the 24-node USNET topology). The distance $L$ between two QT is set to be $160$ km. The required secret-key  rate is considered to be the same for all secure model transmission requests. To streamline the evaluation, the uncertainty $\mathcal{K}$ is reformulated by $\acute{\mathcal{K}}$ as
\begin{equation}
\begin{aligned}
    \acute{\mathcal{K}} = \{(k_1,k_2,\ldots,k_{|\mathcal{R}|}) | &0\leq k_1,k_2,\ldots,k_R<|\Omega|,\\& k_1=k_2=\cdots=k_{|\mathcal{R}|}\},
\end{aligned}\label{eq:34}
\end{equation}
where $|\Omega|$ denotes the total number of scenarios. For each scenario, the probability distribution of secret-key rate requirements obeys a Poisson process, in which the average secret-key rate requirements are set to $\left \lfloor |\Omega|/3 \right \rfloor$. The default number of scenarios is set to $|\Omega|=10$. The default numbers of quantum-secured FEL model transmission requests $|\mathcal{R}|$ are set to 100 and 200 for NSFNET and USNET, respectively. Finally, for the performance evaluation and analysis, the cost values are given in Table~\ref{tb:I}, with a unit representing a normalized monetary unit. As for the proposed learning-based approach, the learning rate is set to $0.003$ for Q-networks and $0.0001$ for policy networks. The temperature coefficient $\tau$ is set at $0.01$. The reward discounting factor is set to $0.95$, and the update coefficient of target Q-networks is set to $0.01$. The number of training epochs is set to $100$. Finally, the size of the replay buffer is set to $20000$ for each learning agent. The SIP model is formulated and solved by Gurobipy 9.5.2 and the DRL algorithms are implemented via PyTorch 1.9.0. 
\begin{table}[h]
\caption{Reserved and on-demand cost values in experimental evaluation (Units)~\cite{cao2021hybrid}}\label{tb:I}
\centering \begin{tabular}{|c|c|c|c|c|c|c|}
\hline
Cost      & $\beta_{Tx}^r$ & $\beta_{Rx}^r$ & $\beta_{KM}^r$ & $\beta_{SI}^r$ & $\beta_{MD}^r$ & $\beta_{Ch}^r$ \\ \hline
Reserved  &           1500     &          2250      & 1200  & 150 & 300 & 1  \\ \hline
On-demand &          6000      &          9000      & 3000  & 500 & 900 & 4  \\ \hline
\end{tabular}
\end{table}




\subsection{Global Model Performance under Different Secret-key Rates}
\begin{figure}[!]
    \centering
    \includegraphics[width=0.8\linewidth]{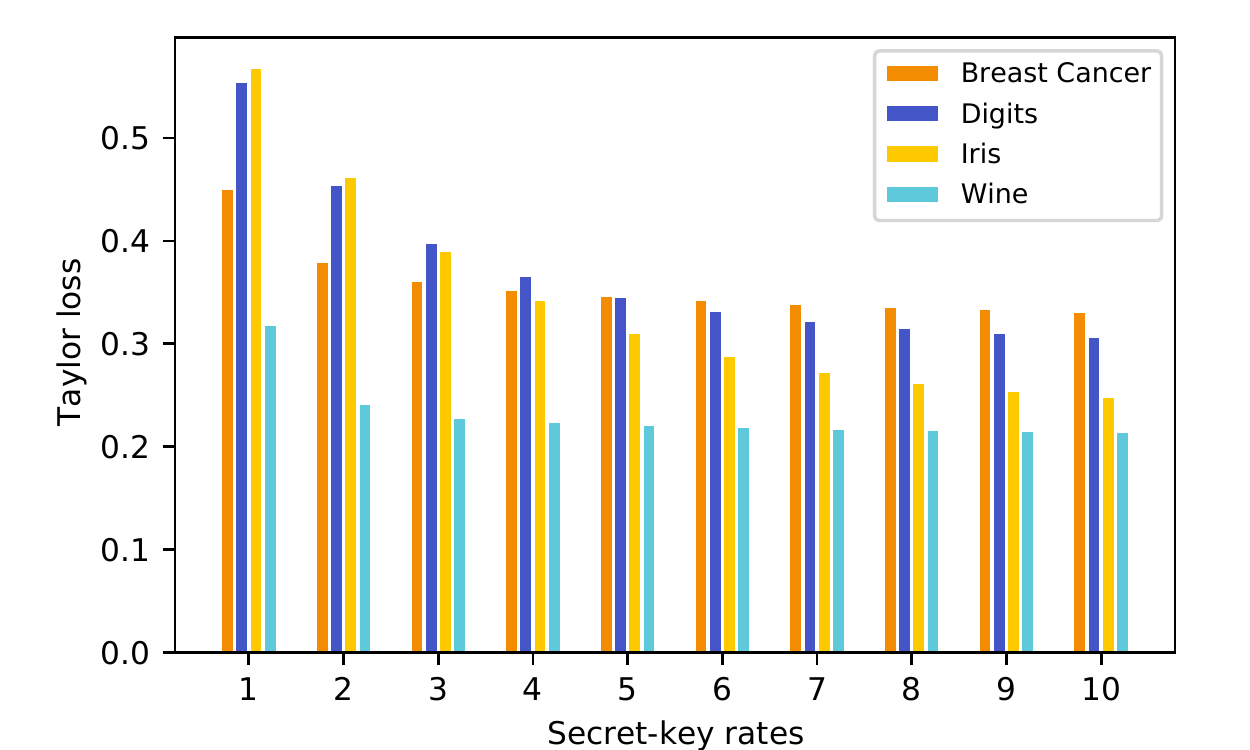}
    \caption{Performance of the global model in quantum-secured FEL systems with various secret-key rates.}
    \label{fig:FEL}
\end{figure}
We conduct our experiments on four datasets from scikit-learn\footnote{\url{https://scikit-learn.org/stable/datasets/toy_dataset.html}}, including breast cancer, digits, iris, and wine. As shown in Fig. \ref{fig:FEL}, the Taylor loss of the global model decreases as the required secret-key rates increase. The model owner can hire more FEL workers if the cluster head has higher secret-key rates to maintain the quantum-secured channels with the FEL workers. Thus, the performance of the global model can be improved by including more knowledge from the datasets of FEL workers. For the complex Digits dataset, the Taylor loss~\cite{hardy2017private} of the classification model decreases from about 0.55 to about 0.3 as the secret-key rates increase. However, the simple wine dataset only requires lower secret-key rates to converge. The reason is that the secret-key rates affect the number of concurrent quantum-secured channels that the cluster heads can maintain, which affects the number of FEL workers in global iterations.

\begin{figure}[t]
\centering
		\subfigure[NSFNET.]{
			\begin{minipage}[t]{0.5\linewidth}
				\includegraphics[width=1\linewidth]{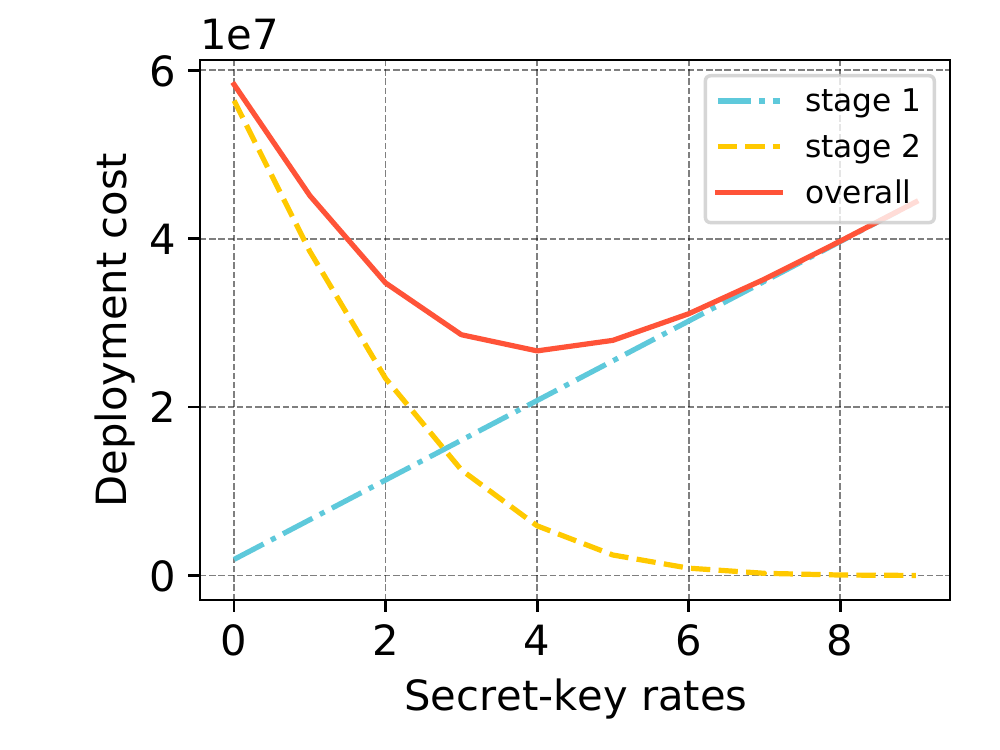}
			\end{minipage}%
		}%
		\subfigure[USNET.]{
			\begin{minipage}[t]{0.5\linewidth}
				
				\includegraphics[width=1\linewidth]{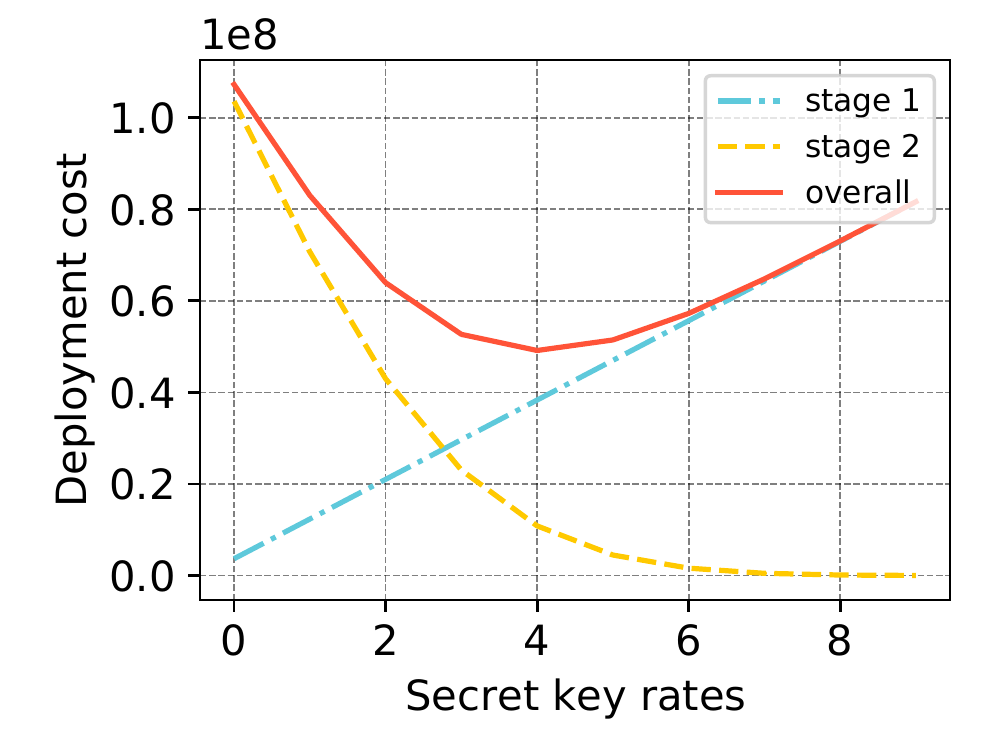}
			\end{minipage}%
		}%
		\caption{The cost structure of SIP for quantum-secured FEL over the NSFNET and the USNET topologies.}
		\label{coststructure}
\end{figure}

\begin{figure}[t]
\centering
		\subfigure[NSFNET.]{
			\begin{minipage}[t]{0.5\linewidth}
				\includegraphics[width=1\linewidth]{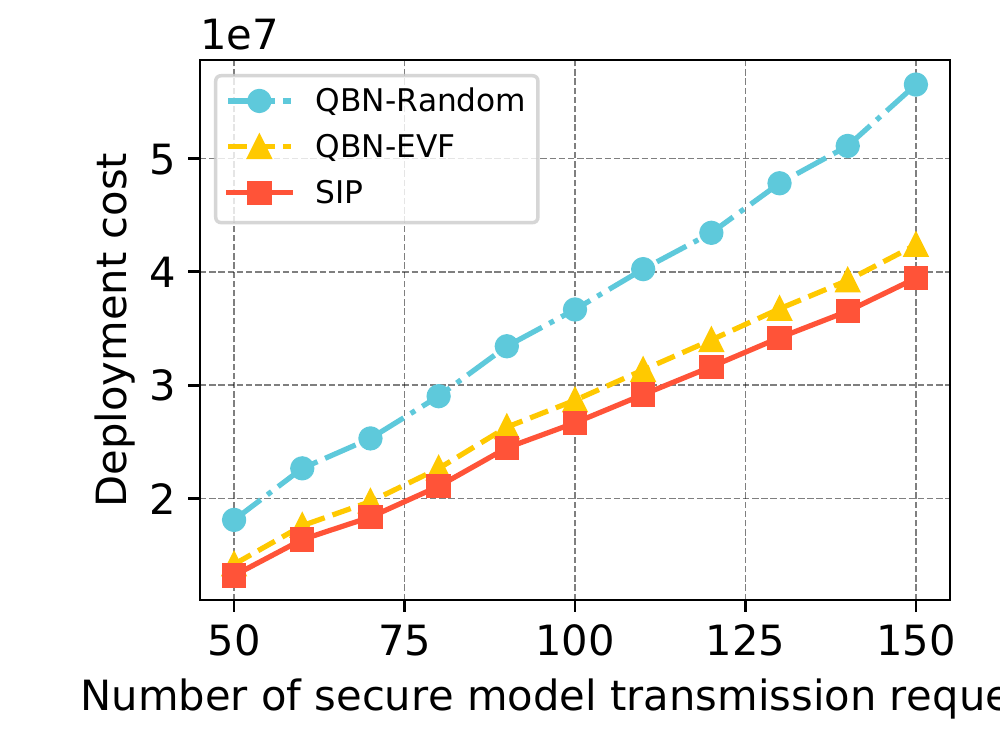}
			\end{minipage}%
		}%
		\subfigure[USNET.]{
			\begin{minipage}[t]{0.5\linewidth}
				
				\includegraphics[width=1\linewidth]{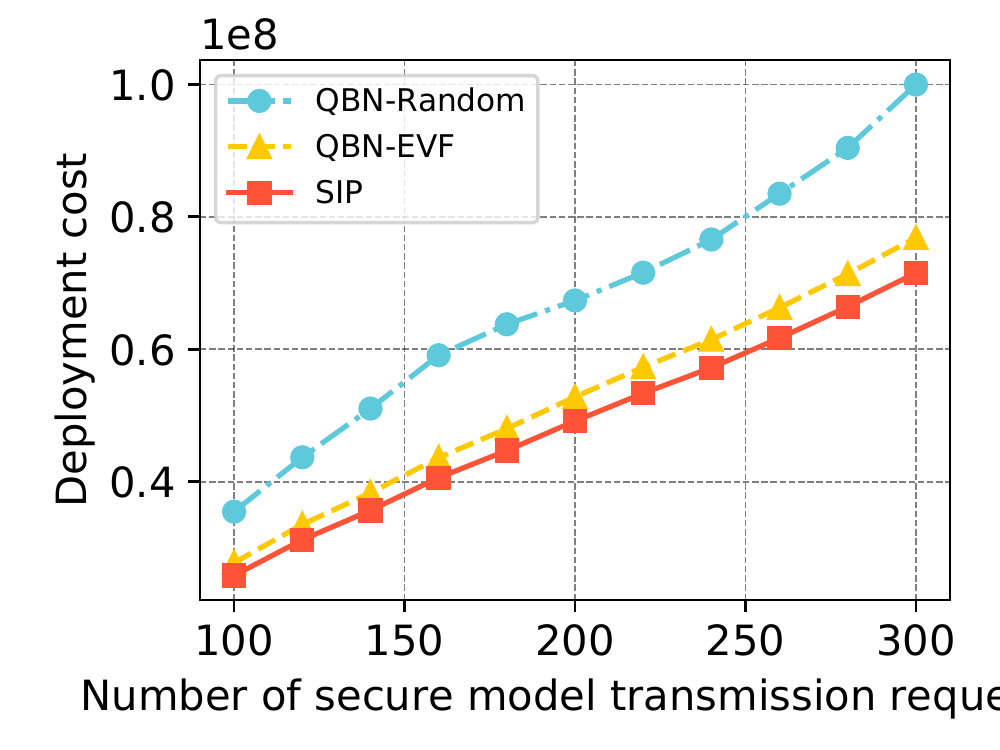}
			\end{minipage}%
		}%
		\caption{Deployment costs of QKD resources versus numbers of quantum-secured FEL model transmission requests.}
		\label{chain}
\end{figure}

\begin{figure}[t]
\centering
		\subfigure[NSFNET.]{
			\begin{minipage}[t]{0.5\linewidth}
				\includegraphics[width=1\linewidth]{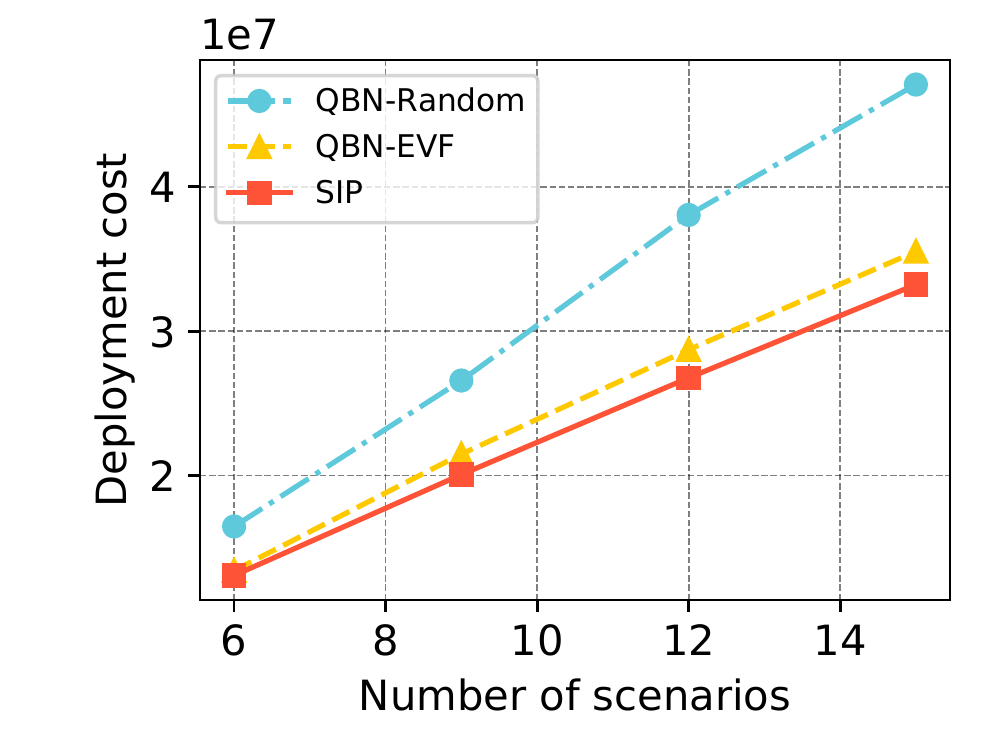}
			\end{minipage}%
		}%
		\subfigure[USNET.]{
			\begin{minipage}[t]{0.5\linewidth}
				
				\includegraphics[width=1\linewidth]{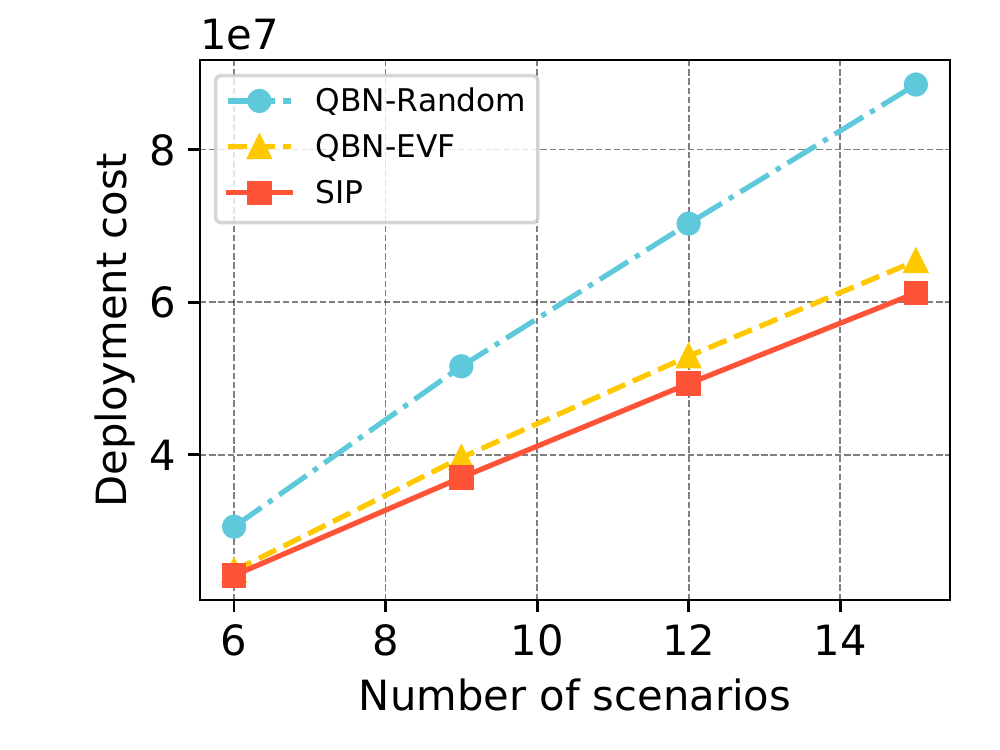}
			\end{minipage}%
		}%
		\caption{Deployment costs of QKD resources versus numbers of scenarios in quantum-secured FEL.}
		\label{scenarios}
\end{figure}

\begin{figure}[t]
\centering
		\subfigure[NSFNET.]{
			\begin{minipage}[t]{0.5\linewidth}
				\includegraphics[width=1\linewidth]{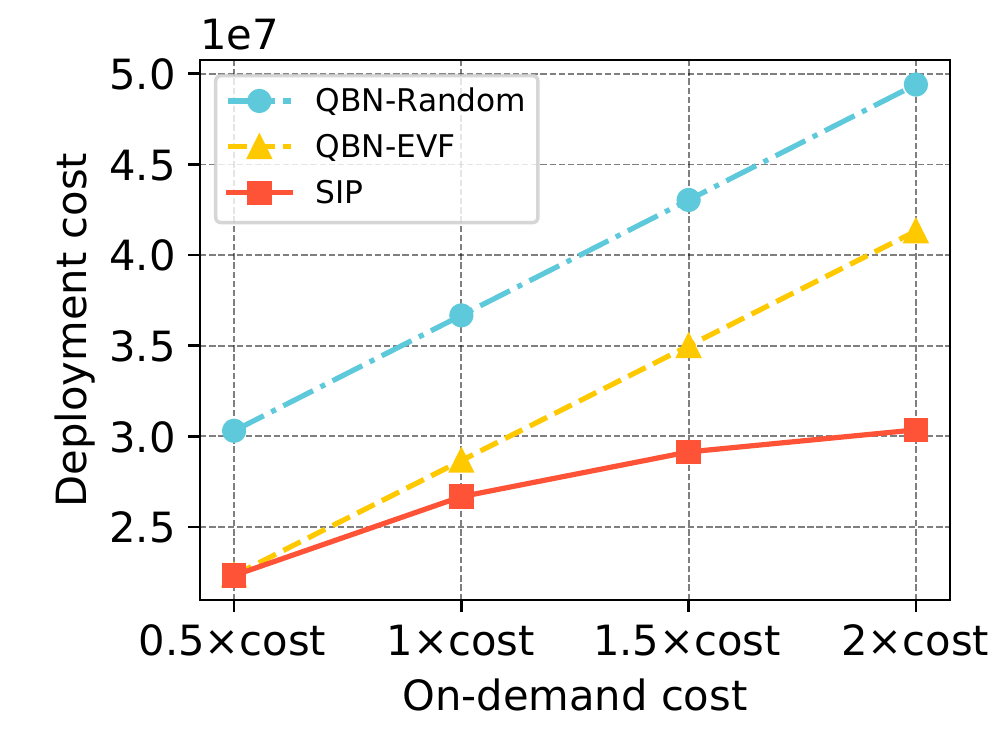}
			\end{minipage}%
		}%
		\subfigure[USNET.]{
			\begin{minipage}[t]{0.5\linewidth}
				
				\includegraphics[width=1\linewidth]{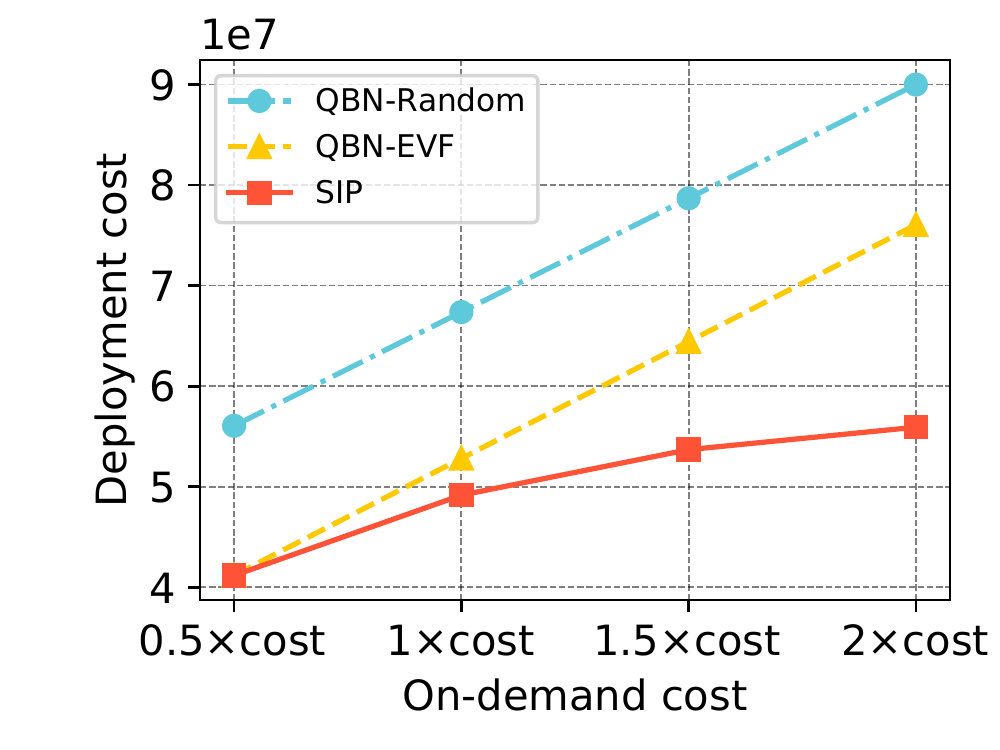}
			\end{minipage}%
		}%
		\caption{Deployment costs of QKD resources versus different on-demand cost schemes}
		\label{ondemand}
\end{figure}

\subsection{Cost Structure Analysis} Initially, the cost structure of the SIP for quantum-secured FEL applications is studied. To simplify the presentation of the cost structure, equal resources are reserved for each QKD request on each link. In Fig. \ref{coststructure}, the costs in the first stage, the second state, and overall are illustrated by varying the number of reserved requests for each link. As expected, the cost of the first stage resource reserve increases with the increase in secret-key rates. However, after the realized demand is known, the costs of the second stage decrease because the reserved QKD wavelengths increase, as the requests require more minor on-demand compensation. Here, the optimized plan can be determined (e.g., when the reserved number of wavelengths equals three, as shown in Fig. \ref{coststructure}). The analysis of the cost structure shows that the optimal solution is not easy to obtain due to the uncertainty of the system in terms of security. For instance, the optimal QKD resource allocation plan is not at the point where the cost of the second stage exceeds the cost of the first stage. Thus, it is necessary to formulate the SIP of the QKD layer in such a way that the deployment cost is optimized.

\subsection{Performance Evaluation under Various Parameters}

In the performance evaluation, the proposed SIP scheme is compared with two baseline schemes, i.e., the QKD backbone networking (QBN) scheme with expected value formulation (QBN-EVF) allocation and the QBN with random (QBN-Random) allocation, presented in~\cite{cao2021hybrid}. In the case of QBN-EVF, the secret-key rates in the first stage are determined by the demanded average, which represents an approximate solution. On the other hand, in the QBN-Random scheme, the values of the decision variables are generated uniformly from zero to $|\Omega|$, indicating a random scheme.

The deployment costs of QBN-Random, QBN-EVF, and SIP for different numbers of quantum-secured FEL model transmission requests are shown first. In Fig. \ref{chain}, we can observe that as the number of secure model transmission requests increases, the deployment cost of all three solutions increases accordingly. Moreover, the difference in deployment cost between the three schemes also increases with the increase in the number of quantum-secured FEL model transmission requests. Specifically, the deployment cost of QBN-EVF is slightly higher than that of SIP, while the deployment cost of QBN-Random is $50\%$ higher than that of SIP. A similar situation arises in the comparison of different numbers of scenarios, as shown in Fig. \ref{scenarios}, namely, the deployment cost of QBN-EVF is slightly higher than that of SIP, while the deployment cost of QBN-Random is twice as expensive as that of SIP. Fig. \ref{ondemand} shows the impact of On-demand costs on the deployment costs of each scenario. For the QBN-Random and QBN-EVF scenarios, the deployment cost increases exponentially with the On-demand cost. In detail, at \textit{0.5$\times$cost}, QBN-Random's deployment cost is 50\% more than SIP's, while QBN-EVF's deployment cost is very close to SIP's. However, when the On-demand cost rises to \textit{2$\times$cost}, the advantage of QBN-EVF decreases, and the deployment cost rises to about 1.5 times that of SIP, while the deployment cost of QBN-Random rises to about twice to SIP's.

\begin{figure}[t]
\centering
		\subfigure[NSFNET.]{
			\centering\begin{minipage}[t]{0.8\linewidth}
				\includegraphics[width=1\linewidth]{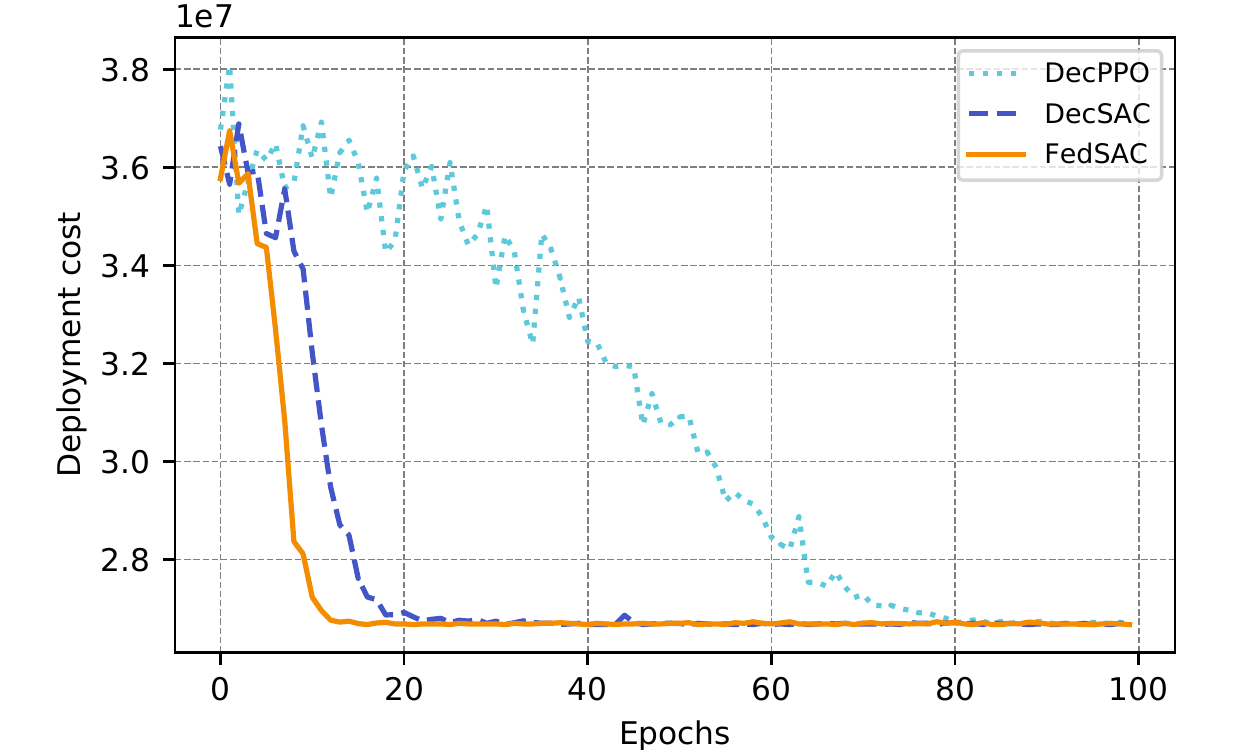}
			\end{minipage}%
		}%

		\subfigure[USNET.]{
			\centering\begin{minipage}[t]{0.8\linewidth}
				\includegraphics[width=1\linewidth]{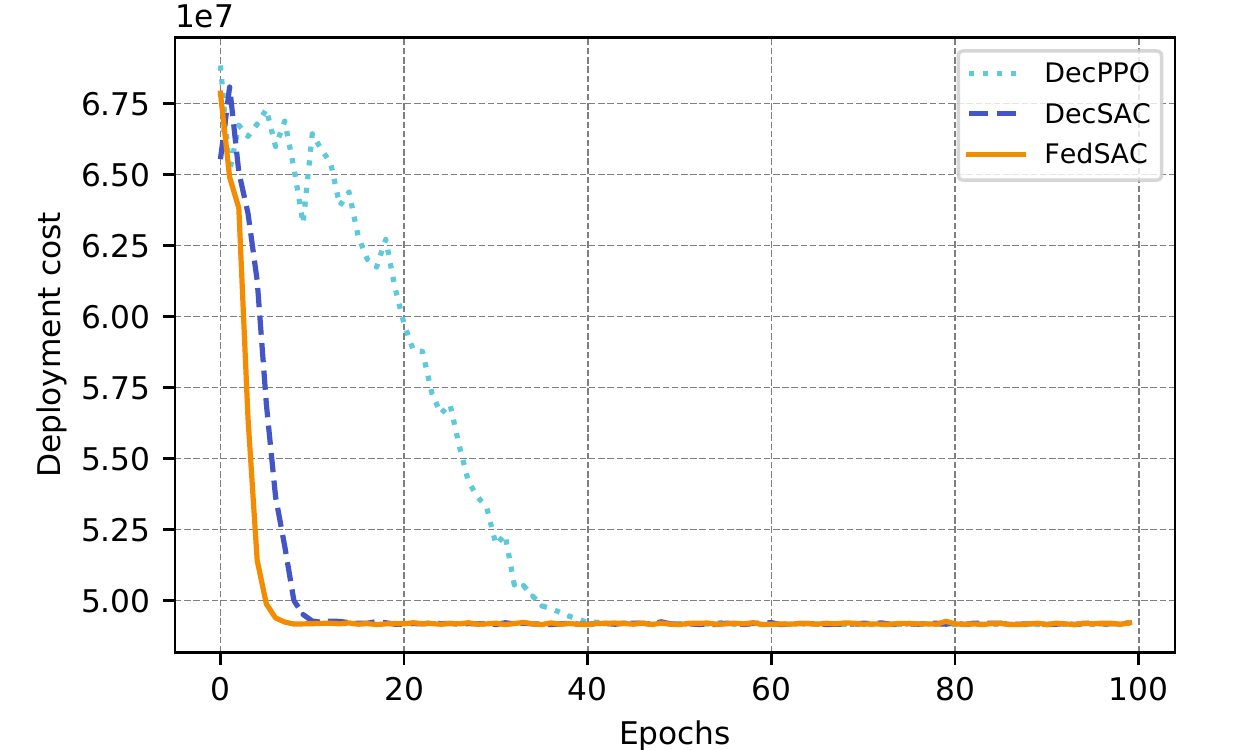}
			\end{minipage}%
		}%
		\caption{Convergence analysis of the proposed scheme.}
		\label{fig:convergence}
\end{figure}
\subsection{Convergence Efficiency of the Proposed Schemes}
In the convergence efficiency analysis of the proposed federated reinforcement learning (FedSAC), the decentralized proximal policy optimization (DecPPO) ~\cite{schulman2017proximal} and the decentralized SAC (DecSAC) schemes are used as the baselines. The DecPPO scheme leverages policy gradient to train the local policy of QKD managers to maximize long-term advantages. The DecSAC scheme adopts soft-Q-learning to optimize the performance of QKD managers while maximizing the entropy of action probability. Both the DecPPO and DecSAC schemes train the QKD managers in a decentralized manner while cannot utilize the incomplete experiences in QKD controllers' replay buffers. As shown in Fig. \ref{fig:convergence}, the proposed scheme and the baseline scheme converge to the optimal QKD resource allocation policy in both network topologies. However, there is a difference in the efficiency of their convergence. For example, the PPO scheme needs about 80 epochs to converge to the optimal QKD resource allocation in the NSFNET topology, while SAC needs about 20 epochs. With the incomplete experience of QKD controllers, FedSAC has a significant improvement in convergence efficiency, requiring only about ten epochs to converge. In the USNET topology, each scheme takes less time to converge, e.g., 40 epochs for PPO, ten epochs for SAC, and five epochs for FedSAC. The reason for this could be that the deployment cost in the USNET topology is slightly higher than in the NSFNET topology, which is a clearer signal for the DRL agents to learn.

\section{Conclusions}\label{sec:conclusions}
 In this paper, we have studied the QKD resource allocation problem for secure federated edge learning. To protect the transmission of FEL models from eavesdropping attacks, we have proposed a hierarchical architecture to facilitate quantum-secured FEL systems. Based on this, we have formulated the optimization of the QKD resource allocation scheme as a SIP model and obtained a cost-effective scheme under uncertainty. To allocate resources in a decentralized and privacy-preserving manner, we proposed a learning-based QKD allocation scheme empowered by federated deep reinforcement learning. The proposed model and scheme have achieved a lower deployment cost of QKD resources compared to the baseline schemes as they can accommodate probabilistic variations in FEL security demand adequately. In the future, we plan to study the resource allocation problem for quantum-secured semantic communication systems in edge networks.

\bibliographystyle{IEEEtran}
\bibliography{myreference}

\end{document}